\documentclass[jounal]{IEEEtran}
\usepackage{cite}
\usepackage{graphicx}
\usepackage{amsmath}
\usepackage{amssymb}
\usepackage{amsthm}
\usepackage{amsfonts}
\usepackage{algorithm}
\usepackage[noend]{algpseudocode}
\usepackage{array}
\usepackage{color}
\usepackage{enumerate}
\usepackage{subfigure}

\makeatletter

\usepackage{algorithm,algpseudocode,float}
\usepackage{lipsum}



\begin{document}

\title{
Asynchronous Activity Detection for Cell-Free Massive MIMO: From Centralized to Distributed Algorithms}


\author{Yang~Li,~
        Qingfeng~Lin,~
        Ya-Feng~Liu,~
        Bo~Ai,~%
        and~Yik-Chung~Wu

\thanks{Y. Li is with Shenzhen Research Institute of Big Data,
Shenzhen 518172, China, and also with State Key Laboratory of Rail Traffic Control and Safety,
Beijing Jiaotong University, Beijing 100044, China
(e-mail: liyang@sribd.cn).}
\thanks{Q. Lin and Y.-C. Wu are with the Department of Electrical and Electronic Engineering, The University of Hong Kong, Hong
Kong (e-mail: qflin@eee.hku.hk; ycwu@eee.hku.hk).}
\thanks{Y.-F. Liu is with the State Key Laboratory of Scientific and Engineering
Computing, Institute of Computational Mathematics and Scientific/Engineering
Computing, Academy of Mathematics and Systems Science, Chinese Academy
of Sciences, Beijing 100190, China (e-mail: yafliu@lsec.cc.ac.cn).}
\thanks{B. Ai is with State Key Laboratory of Rail
Traffic Control and Safety, Beijing Jiaotong University, Beijing 100044, China,
with Peng Cheng Laboratory, Shenzhen 518055, China, and with
Henan Joint International Research Laboratory of Intelligent Networking and
Data Analysis, Zhengzhou University, Zhengzhou 450001, China (e-mail: boai@bjtu.edu.cn).}
}

\maketitle
\begin{abstract}
Device activity detection in the emerging cell-free massive
\textcolor{black}{multiple-input multiple-output (MIMO)} systems
has been recognized as a crucial task in machine-type communications,
in which \textcolor{black}{multiple access points (APs) jointly}
identify the active devices from a large number of potential devices based on the received signals.
\textcolor{black}{Most of the existing works addressing this problem rely on}
the \textcolor{black}{impractical}
assumption that different active devices transmit signals synchronously.
However, in practice, synchronization cannot be guaranteed due to the low-cost oscillators,
which brings \textcolor{black}{additional}
discontinuous and nonconvex constraints to the detection problem.
\textcolor{black}{To address this challenge, this paper reveals an equivalent reformulation to the asynchronous activity detection problem, which facilitates the development of}
a centralized algorithm and a distributed algorithm
that \textcolor{black}{satisfy} the highly nonconvex constraints in
a gentle fashion as the iteration number increases, so that the sequence
\textcolor{black}{generated by the proposed algorithms} can get around
bad stationary points. To reduce the capacity requirements of the fronthauls,
we further design a communication-efficient \textcolor{black}{accelerated}
distributed algorithm.
Simulation results demonstrate that the
proposed centralized \textcolor{black}{and distributed algorithms outperform}
state-of-the\textcolor{black}{-art} approaches,
and the proposed \textcolor{black}{accelerated} distributed algorithm achieves a close detection performance to that of the centralized
algorithm \textcolor{black}{but with a much smaller number of bits to be transmitted on the fronthaul links}.
\end{abstract}
\begin{IEEEkeywords}
Asynchronous activity detection, cell-free massive \textcolor{black}{multiple-input multiple-output (MIMO)}, grant-free random access, Internet-of-Things (IoT), machine-type communications (MTC), nonsmooth and nonconvex optimization.
\end{IEEEkeywords}


\section{Introduction}
As a new paradigm in the fifth-generation and beyond wireless systems,
machine-type communications (MTC) provide efficient random access for
a large number of Internet-of-Things (IoT) devices, of which only a small portion are active at any given time due to the sporadic traffics \cite{Bockelmann2016}. To meet the low-latency requirement in MTC,
a grant-free random access scheme was advocated
in \cite{LiuLiang2018,XM21}, where the devices transmit signals
without the permissions from the access points (APs).
A crucial task during the random access phase is device activity detection,
in which each active device transmits a unique signature sequence so that
the APs \textcolor{black}{could}
identify the active devices from the received signals \cite{LiuLiang2018-1,LiuLiang2018-2}.

\textcolor{black}{However}, due to the large number of devices but limited coherence time,
the signature sequences have to be nonorthogonal, and hence the
interference among different devices makes device activity detection challenging.
Moreover, since the IoT devices are commonly equipped with low-cost oscillators,
the transmissions of different active devices cannot be perfectly synchronized,
\textcolor{black}{which brings} an additional challenge to the task of device activity detection.

While asynchronous transmissions are common for IoT devices,
the existing studies \textcolor{black}{on device activity detection}
focus more on \textcolor{black}{the synchronous case}
\cite{Chen2018,Ding19,Ke20,Chen2019,Senel2018,Yuan20_2,Yuan20_5,Jiang2020,Mei21,Ai22,Liu2018,He2018,Li2019,Shao2020,LiTing22,
Haghighatshoar2018,chenIcc,icasspW,Wang21,Chen2021,Lin2022,Fengler2021,Chen2020,Ke2021,Shao2020-2,Ganesan2021},
which can be roughly divided into two lines of research.
\textcolor{black}{In the first line of research,}
by exploiting the sporadic traffics, compressed sensing (CS)
based methods have been widely studied
\cite{Chen2018,Ding19,Ke20,Chen2019,Senel2018,Yuan20_2,Yuan20_5,Jiang2020,Mei21,Ai22,Liu2018,He2018,Li2019,Shao2020,LiTing22}. In particular, approximate message passing (AMP) was applied
to jointly estimate the activity status and the instantaneous channels in \cite{Chen2018,Ding19,Ke20,Chen2019},
and was further extended to include data detection in \cite{Senel2018,Yuan20_2,Yuan20_5,Jiang2020,Mei21}.
Besides, Bayesian sparse recovery \cite{Ai22}
and sparse optimization \cite{Liu2018,He2018,Li2019,Shao2020,LiTing22}
have also been investigated in the literature.
\textcolor{black}{Instead} of performing joint activity detection
and channel estimation, another line of research
called covariance approach identifies the active devices
without estimating the instantaneous channel \cite{Haghighatshoar2018,chenIcc,icasspW,Wang21,Chen2021,Lin2022}.
This approach exploits the statistical properties of the channel based on
the sample covariance matrix.
Compared with the CS based methods, analytical results have shown that
the covariance approach \textcolor{black}{can achieve a}
better detection performance with \textcolor{black}{a}
much shorter signature sequence length \cite{Fengler2021,Chen2020}.

Recently, cell-free massive \textcolor{black}{multiple-input multiple-output (MIMO)} has been recognized as an efficient
architecture for providing uniformly high data rates, in which
all the APs are connected to a central processing
unit (CPU) via fronthaul links for joint
signal processing. Cell-free massive MIMO has no ``cell boundaries'',
and hence overcomes the inter-cell interference.
As compared to the traditional network architecture, recent studies have shown that
cell-free massive MIMO can provide a better activity detection performance for MTC
using the CS based method \cite{Ke2021} or the covariance approach \cite{Shao2020-2,Ganesan2021}.

While the above existing works exemplify the possibility of
device activity detection,
they are designed under the
perfect synchronization assumption.
However, in practice, due to low-cost oscillators in IoT devices,
synchronous transmissions among different devices cannot be guaranteed \cite{YC11}.
Even though network synchronization algorithms \cite{Du2013,Luo2013} can be executed before activity detection, synchronization errors still exist.
This makes the received signature sequences
in actual scenarios largely different from those assumed in existing works.
Consequently, the above activity detection methods based on
the synchronous assumption \textcolor{black}{suffer significant degradation when applied in} asynchronous activity detection.
To address this issue,
\textcolor{black}{the} work \cite{Wang22} introduced $\ell_0$-norm constraints into the covariance based optimization problem, \textcolor{black}{equivalently making}
each active device transmit only one
effective signature sequence from different possible delays
in each transmission. To tackle the $\ell_0$-norm constrained problem,
\cite{Wang22} further proposed a block coordinate descent (BCD) algorithm,
which shows significant performance improvement compared with a CS based method \cite{LiuLiang21}.

\textcolor{black}{Unfortunately, the current solution in \cite{Wang22} faces a major challenge, which comes from the fact that enforcing these discontinuous and nonconvex $\ell_0$-norm
constraints within each iteration of the BCD algorithm
may cause} the solution to \textcolor{black}{get} stuck at bad stationary points\textcolor{black}{\cite{Foucart2009SparsestSO,Soubies15}}, which will degrade the detection performance.
To tackle \textcolor{black}{this} challenge, this paper proposes a novel equivalent penalized reformulation for the original asynchronous activity detection problem. We prove that these two problems are \textcolor{black}{equivalent} in the sense that their global optimal solutions
are identical under mild conditions.
We further propose an efficient centralized detection algorithm to solve the reformulated problem to a stationary point, which is also proved to be a stationary point
of the original problem\footnote{\textcolor{black}{
In this paper, the stationary point of a problem with nonconvex $\ell_0$-norm constraints
is more rigorously a B-stationary point, at which the
directional derivatives along any direction within its tangent cone are non-negative \cite{facchinei2007finite}.}}.
Instead of enforcing the highly nonconvex $\ell_0$-norm constraints within each iteration,
the proposed centralized algorithm guarantees these constraints to be satisfied \textcolor{black}{progressively}
as the iteration number increases. Therefore, the sequence \textcolor{black}{generated by the proposed algorithms} will not
\textcolor{black}{get} stuck at bad stationary points \textcolor{black}{caused} by the highly nonconvex $\ell_0$-norm constraints.
Simulation results show that the detection performance of the proposed centralized algorithm is much better than that of state-of-the-art approaches \cite{Wang22}.

While the proposed centralized algorithm outperforms the existing approaches,
it is totally executed at the CPU based on the received signals \textcolor{black}{collected} from all APs.
Thus, the computational \textcolor{black}{burden would be heavy} especially when the network size becomes large.
To reduce the computational cost at the CPU, it is more appealing to design a distributed algorithm
\textcolor{black}{in which part of the computations can be} performed at the APs \cite{Emil20,He21}.
\textcolor{black}{Going towards this direction},
we further propose a distributed detection algorithm,
which is executed at both the APs and the CPU.
Specifically, each AP performs a local detection for the devices and then sends its detection results to the CPU for further processing. In this way, the computations are balanced on various parts of the network. Moreover, the proposed distributed algorithm is also proved to converge to a stationary point
with the same solution quality to that of the centralized algorithm.

Notice that both the centralized and distributed algorithms
\textcolor{black}{require} communication overheads for exchanging information between the APs and the CPU.
Since the fronthauls are capacity-limited, the exchanged contents have to be compressed before
being transmitted \cite{Manijeh19,Manijeh21}. However, the compression error will in turn degrade the detection performance.
Therefore, it is desirable to design a communication-efficient algorithm that can reduce the capacity requirements of the fronthauls while still achieving satisfactory detection performance.
To this end, we further propose a heuristic scheme to modify the distributed algorithm such that the convergence is accelerated and the exchanged variables appear in a much smaller dynamic range.
Simulation results demonstrate that the accelerated distributed algorithm
achieves a close detection performance to that of the centralized algorithm
with only $1$ iteration, and its required number of bits transmitted on the fronthaul links
is also much smaller than
that of the centralized algorithm.

The remainder of this paper is organized as follows. System model and problem formulation
are presented in Section II.
\textcolor{black}{A centralized algorithm and a distributed algorithm are proposed
under perfect fronthaul links in Section III and Section IV, respectively.
A communication-efficient scheme is presented for practical capacity-limited fronthaul links in Section V.}
Finally, Section VII concludes the paper.

\section{System Model and Problem Formulation}
\subsection{System Model}
Consider an uplink cell-free massive MIMO system with $M$
APs and $K$ IoT devices arbitrarily and independently distributed in the network.
Each AP is equipped with $N$ antennas and
each device is equipped with a single antenna.
All the $M$ APs are connected to a CPU
via fronthaul links, so that the received signal from each AP
can be collected and jointly processed
at the CPU.
We adopt a quasi-static block-fading channel model, where
the \textcolor{black}{channels between devices and APs
remain} constant within each coherence block,
but may vary among different coherence blocks.
Let $\sqrt{g_{k,m}}\mathbf{h}_{k,m}$ denote the channel from the
$k$-th device to the $m$-th AP, where $\sqrt{g_{k,m}}$ and $\mathbf{h}_{k,m}\in\mathbb{C}^{N}$
are the large-scale and small-scale
fading components, respectively.
\textcolor{black}{In many practical deployment scenarios, the devices are stationary,
so their large-scale fading channels are fixed and can be obtained in advance using conventional
channel estimation methods \cite{Mossberg2006,Wang2018}.
In this paper, the large-scale fading channels are assumed to be known as in \cite{Senel2018,Chen2021,Ganesan2021,Wang21}. Moreover,
we consider that there are many objects in the environment that scatter the signal before arriving at each AP, so each entry in $\mathbf{h}_{k,m}$
can be well-modeled by $\mathcal{CN}\left(0, 1\right)$ \cite{BernardSklar97}.}
Due to the sporadic traffics of MTC, only a small portion of the $K$ devices are active in each
coherence block. If the $k$-th device is active, the activity status is denoted as $a_k=1$
(otherwise, $a_k=0$).

To detect the activities of the IoT devices, we assign each device a unique
signature sequence $\bar{\mathbf{s}}_k\in\mathbb{C}^{L}, \forall k=1,\ldots,K$, where
$L$ is the length of the signature sequence\textcolor{black}{\footnote{\textcolor{black}{
The length of the signature sequence $L$ is usually fixed within a deployment period of the network.
The value of $L$ realizes a trade-off between the detection performance and the computational complexity.
A larger $L$ will improve the detection performance but increase the computational complexity of the detection algorithms.}}}.
The signature sequences of all the $K$ devices
are assumed to be known.
Since the devices are commonly equipped with low-cost local oscillators,
the signature sequences of different devices \textcolor{black}{may not be transmitted synchronously}.
In particular, we assume that the $k$-th device transmits its
signature sequence with an unknown delay of $t_k\in\{0,\ldots,T\}$ symbols,
where the maximum delay $T$ is known\textcolor{black}{\footnote{\textcolor{black}{
The maximum delay $T$ depends on the symbol duration and the oscillators equipped on the device.
For instance, when the symbol duration is $5$ $\mu$s (when the signal bandwidth is $200$ kHz),
and the oscillators result in a maximum delay of $20$ $\mu$s,
the value of $T$ is $4$ symbols.
In the simulations of \cite{Wang22,LiuLiang21},
$T$ is set as $4$ and $5$ symbols, respectively.
To compare the detection performance under different $T$,
we vary it from $0$ to $8$ in Section VI.}}}.
With the $k$-th device transmitting its signature sequence
at the $(t_k+1)$-th symbol duration,
its effective signature sequence can be expressed as
\begin{equation}\label{effective}
\mathbf{s}_{k,t_k} = \big[\underbrace{0,\ldots,0}_{t_k},\bar{\mathbf{s}}_k^\mathrm{T},\underbrace{0,\ldots,0}_{T-t_k}
\big]^\mathrm{T}\in\mathbb{C}^{L+T},~~\forall k=1,\ldots,K.
\end{equation}
Consequently, the received signal over the $L+T$ symbol durations
at each AP can be written as
\begin{equation}\label{received}
\mathbf{Y}_m =
\sum_{k=1}^{K} a_k\sqrt{p_kg_{k,m}}\mathbf{s}_{k,t_k}
\mathbf{h}_{k,m}^\mathrm{T}+\mathbf{W}_m,~~\forall m=1,
\ldots,M,
\end{equation}
where $p_k$ is the transmit power of the $k$-th device
and the elements of $\mathbf{W}_m\in\mathbb{C}^{\left(L+T\right)\times N}$ are independent and identically distributed (i.i.d.) Gaussian noise
at the $m$-th AP following $\mathcal{C} \mathcal{N}\left(0, \sigma_m^{2} \right)$ with $\sigma_m^{2}$ being the noise variance.
\textcolor{black}{In \eqref{received}, the transmit power $p_k$ of each device can be different \cite{Sadeghabadi18}.
In order to reduce the channel gain variations among different devices,
$p_k$ can be controlled based on the large-scale fading component
to its dominant AP,
which is the AP with the largest channel gain \cite{Senel2018,Ganesan2021}.}

To express the received signal in \eqref{received}
more compactly, we denote an indicator of the device activity and delay
for each device as
\begin{eqnarray}\label{variable}
b_{k,t} =
\begin{cases}
1,~~&\text{if}~~a_k=1~~\text{and}~~t=t_k,\cr
0,~~&\text{otherwise},
\end{cases}\nonumber\\
\forall k=1,\ldots,K,~~\forall t=0,\ldots,T,
\end{eqnarray}
which means that $b_{k,t}=1$ if and only if device $k$ is active with a delay of
$t$ symbol durations.
Since there is at most one possible delay for each device in each transmission,
we have $\textcolor{black}{a_k=}\sum\limits_{t=0}^{T} b_{k, t} \in\{0,1\}, \forall k=1,\ldots,K$.
Thus, the received signal in \eqref{received} can be rewritten as
\begin{eqnarray}\label{received2}
\mathbf{Y}_m &=&
\sum_{k=1}^{K}\sum_{t=0}^{T} b_{k,t}\sqrt{p_kg_{k,m}}\mathbf{s}_{k,t}
\mathbf{h}_{k,m}^\mathrm{T}+\mathbf{W}_m\nonumber\\
&=&\sum_{k=1}^{K} \mathbf{S}_{k}\mathbf{B}_{k}
\mathbf{G}_{k,m}^{\frac{1}{2}}\mathbf{H}_{k,m}+\mathbf{W}_m,
~~\forall m=1,\ldots,M,\nonumber\\
\end{eqnarray}
where
$\mathbf{S}_k\triangleq\left[\mathbf{s}_{k,0}, \ldots, \mathbf{s}_{k,T}\right]\in\mathbb{C}^{(L+T)\times (T+1)}$ is the effective signature matrix of device $k$,
$\mathbf{B}_k\triangleq\text{diag}\left\{b_{k,0}, \ldots, b_{k,T}\right\}$,
$\mathbf{G}_{k,m}\triangleq\text{diag}\big\{\underbrace{p_kg_{k,m}, \ldots, p_kg_{k,m}}_{T+1}\big\}$,
and $\mathbf{H}_{k,m}\triangleq\big[\underbrace{\mathbf{h}_{k,m}, \ldots, \mathbf{h}_{k,m}}_{T+1}\big]^\mathrm{T}\in\mathbb{C}^{(T+1)\times N}$.

\subsection{Problem Formulation}
Mathematically, the asynchronous activity detection is equivalent to detecting each $b_{k, t}\in \{0,1\}$, which includes both the information of the device activity and its transmission delay (if it is active). Specifically, if $b_{k, t}$ is detected as $1$, we believe that the $k$-th device should be active with a delay of $t$ symbol durations.

For this purpose, we treat $\left\{\mathbf{S}_k\right\}_{k=1}^K$, $\left\{\mathbf{B}_k\right\}_{k=1}^K$, $\left\{\mathbf{G}_{k,m}\right\}_{k=1,m=1}^{k=K,m=M}$ in \eqref{received2} as deterministic, and treat the small-scale fading channel matrices $\left\{\mathbf{H}_{k,m}\right\}_{k=1,m=1}^{k=K,m=M}$ and the noise $\left\{\mathbf{W}_m\right\}_{m=1}^M$ as complex Gaussian random variables.
Consequently, \textcolor{black}{being a} linear combination of $\left\{\mathbf{H}_{k,m}\right\}_{k=1,m=1}^{k=K,m=M}$ and
$\left\{\mathbf{W}_m\right\}_{m=1}^M$, the received signal $\mathbf{Y}_m$ in
\eqref{received2} is also complex Gaussian distributed.
In particular, with $\mathbf{y}_{m,n}$ denoting
the $n$-th column of $\mathbf{Y}_m$, we have
$\mathbf{y}_{m,n}\sim\mathcal{C} \mathcal{N}\left(\mathbf{0}, \mathbf{C}_m\right)$, where
\begin{eqnarray}\label{covariance}
\mathbf{C}_m&=&\mathbb{E}\left[\mathbf{y}_{m,n}\mathbf{y}_{m,n}^\mathrm{H}\right]\nonumber\\
&=&\sum_{k=1}^{K} \mathbf{S}_{k}\mathbf{B}_{k}
\mathbf{G}_{k,m}^{\frac{1}{2}}\mathbf{E}\mathbf{G}_{k,m}^{\frac{1}{2}}\mathbf{B}_{k}\mathbf{S}_{k}^\mathrm{H}+\sigma_m^{2}\mathbf{I}_{L+T}\nonumber\\
&=&\sum_{k=1}^{K}\sum_{t=0}^{T} b_{k,t}
p_kg_{k,m}\mathbf{s}_{k,t}\mathbf{s}_{k,t}^\mathrm{H}+\sigma_m^{2}\mathbf{I}_{L+T},\nonumber\\
&&\forall m=1,\ldots,M.
\end{eqnarray}
\textcolor{black}{In \eqref{covariance},
$\mathbf{E}=\mathbb{E}\left[\mathbf{h}_{k,m,n}\mathbf{h}_{k,m,n}^\mathrm{H}\right]$ is
an all-one matrix, where $\mathbf{h}_{k,m,n}$ denotes the $n$-th column of $\mathbf{H}_{k,m}$.}
The last equality in
\eqref{covariance} holds since there is at most one non-zero entry in each diagonal matrix $\mathbf{B}_k$.
With $\mathbf{b}_k\triangleq\left[b_{k, 0}, \ldots, b_{k, T}\right]^{\mathrm{T}}$
denoting the diagonal entries of $\mathbf{B}_k$
and $\mathbf{b}\triangleq\left[\mathbf{b}_{1}^{\mathrm{T}}, \ldots, \mathbf{b}_{K}^{\mathrm{T}}\right]^{\mathrm{T}}$,
we can estimate $\mathbf{b}$
by maximizing the likelihood function
\begin{eqnarray}\label{likehood}
&&p\left(\left\{\mathbf{Y}_m\right\}_{m=1}^M;\mathbf{b}\right)\nonumber\\
&=& \prod_{m=1}^{M}\prod_{n=1}^{N} p\left(\mathbf{y}_{m,n};\mathbf{b}\right)\nonumber\\
&=&\prod_{m=1}^{M}\frac{1}{\vert\pi\mathbf{C}_m\vert^N}
\exp\left(-\text{Tr}\left(\mathbf{C}_m^{-1}\mathbf{Y}_m\mathbf{Y}_m^\mathrm{H}\right)\right)\textcolor{black}{.}
\end{eqnarray}
\textcolor{black}{Further considering the constraints of $\mathbf{b}$,}
the optimization problem \textcolor{black}{is given by}
\begin{subequations}\label{opt1}
\begin{align}
\begin{split}\label{opt1cost}
\min _{\mathbf{b}}~~
&\sum_{m=1}^M\left(\log\vert\mathbf{C}_m\vert+\frac{1}{N}\text{Tr}\left(\mathbf{C}_m^{-1}\mathbf{Y}_m\mathbf{Y}_m^\mathrm{H}\right)\right),
\end{split}\\
\begin{split}\label{opt1c1}
\text{s.t.}\ ~~
&\mathbf{b}\in\left[0,1\right]^{K(T+1)},~~~~~~~~~~~~~~~~~~~~~~~~~~~~
\end{split}\\
\begin{split}\label{opt1c2}
&\left\Vert\mathbf{b}_k\right\Vert_0\leq1,~~\forall k=1,\ldots,K,
\end{split}
\end{align}
\end{subequations}
where \eqref{opt1c1}
is a continuous relaxation of the binary constraint $\mathbf{b}\in\left\{0,1\right\}^{K(T+1)}$,
and \eqref{opt1c2} is \textcolor{black}{due to}
at most one possible delay
for each device in each transmission, i.e., $\sum\limits_{t=0}^{T} b_{k, t} \in\{0,1\}$.

\vspace{2mm}
\textcolor{black}{
\emph{Remark 1:}
The constraint \eqref{opt1c1}
is a reasonable relaxation due to three aspects.
Firstly, the constraint \eqref{opt1c2}
already guarantees that there is at most $1$ non-zero entry in each $\mathbf{b}_k$, which
means that at least $KT$ entries in $\mathbf{b}$ are guaranteed to locate at the boundary $0$.
Secondly, minimizing the first term of \eqref{opt1cost},
$\sum_{m=1}^M\log\vert\sum_{k=1}^{K}\sum_{t=0}^{T} b_{k,t}
p_kg_{k,m}\mathbf{s}_{k,t}\mathbf{s}_{k,t}^\mathrm{H}+\sigma_m^{2}\mathbf{I}_{L+T}\vert$,
has the effect of minimizing the rank of $\sum_{k=1}^{K}\sum_{t=0}^{T}b_{k,t}
p_kg_{k,m}\mathbf{s}_{k,t}\mathbf{s}_{k,t}^\mathrm{H}$ \cite{Fazel2003},
which will enforce most of $\mathbf{b}_k$ to be all-zero vectors as
demonstrated by the simulations in \cite{Haghighatshoar2018,chenIcc}.
Thirdly, after solving problem \eqref{opt1}, for the very few entries $b_{k,t}$
that are not at the boundary of $[0,1]$, we
adopt a threshold $\gamma\in[0,1]$ to recover the binary variable as
$\hat{b}_{k,t}=\mathbb{I}(b_{k,t}>\gamma)$.
By varying the threshold $\gamma$ in $[0, 1]$, we can achieve a good trade-off between the probability of missed detection and the probability of false alarm as shown in Section~VI.}

\vspace{2mm}
\textcolor{black}{
\emph{Remark 2:}
When the devices are equipped with multiple antennas,
the signature sequence of each device can be transmitted with the help of beamforming.
Specifically, let $N_{\text{r}}\geq1$ and $N_{\text{t}}\geq1$ denote the numbers of antennas
at each AP and each device, respectively.
Let $\tilde{\mathbf{H}}_{m,k}\in\mathbb{C}^{N_{\text{r}}\times N_{\text{t}}}$
denote the small-scale Rayleigh fading channel from the $k$-th device to the $m$-th AP.
Let $\mathbf{w}_k\in\mathbb{C}^{N_{\text{t}}}$ with $\Vert\mathbf{w}_k\Vert_2=1$
denote the beamforming vector at the $k$-th device. Consequently,
an effective channel from the $k$-th device to the $m$-th AP can be written as $\tilde{\mathbf{h}}_{m,k}=\tilde{\mathbf{H}}_{m,k}\mathbf{w}_k$.
Since $\Vert\mathbf{w}_k\Vert_2=1$ and each column of $\tilde{\mathbf{H}}_{m,k}$ follows i.i.d. $\mathcal{CN}(\mathbf{0}, \mathbf{\mathbf{I}}_{N_{\text{r}}})$,
we also have $\tilde{\mathbf{h}}_{m,k}\sim\mathcal{CN}(\mathbf{0}, \mathbf{\mathbf{I}}_{N_{\text{r}}})$,
which has the same probability distribution as that of the
small-scale channel $\mathbf{h}_{m,k}$ for $N_{\text{t}}=1$.
Therefore, with $\tilde{\mathbf{h}}_{m,k}$ replacing $\mathbf{h}_{m,k}$,
the detection problem for $N_{\text{t}}>1$ can still be formulated as problem \eqref{opt1},
and hence can be solved by the proposed algorithms in the following sections.}

\subsection{Penalized Reformulation of Problem \eqref{opt1}}
Problem \eqref{opt1} is challenging to solve due to the discontinuous and nonconvex $\ell_0$-norm constraints in \eqref{opt1c2} caused by the asynchronous \textcolor{black}{transmissions}.
Recently, \textcolor{black}{the} work \cite{Wang22} proposed a BCD algorithm
for single-cell asynchronous activity detection.
This algorithm enforces the $\ell_0$-norm constraints within the optimization process,
which can guarantee its feasibility to \eqref{opt1c2}.
However, since $\ell_0$-norm is highly nonconvex,
enforcing these hard constraints during the iterations may make
it \textcolor{black}{easily get} stuck at bad stationary points,
and hence degrades the detection performance.

Instead of explicitly enforcing the $\ell_0$-norm constraints in \eqref{opt1c2}, we transform problem
\eqref{opt1} into: 
\begin{eqnarray}\label{opt2}
\min_{\mathbf{b}\in\left[0,1\right]^{K(T+1)}}
&&\sum_{m=1}^M\left(\log\vert\mathbf{C}_m\vert+\frac{1}{N}\text{Tr}\left(\mathbf{C}_m^{-1}\mathbf{Y}_m\mathbf{Y}_m^\mathrm{H}\right)\right)\nonumber\\
&&+\rho\sum_{k=1}^K \left(\sum_{t=0}^Tb_{k,t}-\max_{t\in\{0,\ldots,T\}}b_{k,t}\right),
\end{eqnarray}
where $\rho>0$ is a penalty parameter to penalize the violation of \eqref{opt1c2}.
The following theorem establishes the equivalence between
the original problem \eqref{opt1} and the penalized problem \eqref{opt2}.

\newtheorem{theorem}{Theorem}
\begin{theorem}\label{equivalence}
The two problems \eqref{opt1} and \eqref{opt2} are equivalent in the sense
that there exists a finite $\rho^*<\infty$ such that for any $\rho>\rho^*$:
\textcolor{black}{
\begin{enumerate}[1)]
\item any stationary point of problem \eqref{opt2} must be a stationary point of problem \eqref{opt1};
\item the global optimal solutions of the two problems are identical.
\end{enumerate}}
\end{theorem}
\begin{proof}
See Appendix \ref{proof of equivalence}.
\end{proof}

\vspace{2mm}

The equivalence implies that the \textcolor{black}{solution of the}
original problem \eqref{opt1} can be accomplished via solving the penalized problem \eqref{opt2}, which is easier to handle since problem \eqref{opt2} has only simple box constraints.
\textcolor{black}{In contrast to the original problem \eqref{opt1}, which
is very likely to get stuck at bad stationary points,
problem \eqref{opt2} has a larger feasible set than that of problem \eqref{opt1}
after removing the $\ell_0$-norm constraints},
making an iterative algorithm easier
to get around bad stationary points.
More importantly, this benefit comes without sacrificing the solution quality.
\textcolor{black}{In particular,} theorem~1 shows that when $\rho$ is sufficiently large,
as long as we solve problem \eqref{opt2} to a stationary point, it must also be
a stationary point \textcolor{black}{(and thus a feasible point)}
of problem \eqref{opt1}.
This means that an efficient algorithm for solving problem \eqref{opt2} can
satisfy \eqref{opt1c2} in a gentle fashion when approaching
a better stationary point.

\section{Centralized Algorithm for Asynchronous Activity Detection}
In this section, 
we propose an efficient algorithm for solving problem \eqref{opt2} to a stationary point,
which is also a stationary point of problem \eqref{opt1}.
This proposed algorithm is executed at the CPU for centralized detection,
which also provides a performance \textcolor{black}{reference} under ideal fronthauls.


Notice that the penalized problem \eqref{opt2} is a nonsmooth problem,
where the cost function is in the form of a differentiable part plus
a non-differentiable part, i.e.,
$-\rho\sum\limits_{k=1}^K \max\limits_{t\in\{0,\ldots,T\}}b_{k,t}$.
This type of nonsmooth problem can be tackled
by \textcolor{black}{the} proximal gradient method \cite{Parikh2014}. To solve problem \eqref{opt2},
\textcolor{black}{the} proximal gradient method adopts the following update
at the $i$-th iteration:
\begin{eqnarray}\label{PG}
\mathbf{b}^{(i)}
= \arg \min_{\mathbf{b} \in [0,1]^{K(T+1)} }&&\frac{1}{2 \eta_{i}}\left\Vert \mathbf{b}-\left(\mathbf{b}^{(i-1)}-\eta_i\mathbf{d}^{(i-1)}\right)\right\Vert_{2}^{2}\nonumber\\
&&-\rho\sum_{k=1}^K\max_{t\in\{0,\ldots,T\}}b_{k,t},
\end{eqnarray}
where $\eta_i$ is the step size and
$\mathbf{d}^{(i-1)}$ is
the gradient of the differentiable part \textcolor{black}{of} \textcolor{black}{problem} \eqref{opt2}
with respect to \textcolor{black}{$\mathbf{b}$ at} $\mathbf{b}^{(i-1)}$.
In particular, \textcolor{black}{the differentiable part of problem \eqref{opt2} is}
\begin{eqnarray}\label{smooth}
G_0(\mathbf{b})&\triangleq&\sum_{m=1}^M\left(\log\vert\mathbf{C}_m\vert+\frac{1}{N}\text{Tr}\left(\mathbf{C}_m^{-1}\mathbf{Y}_m\mathbf{Y}_m^\mathrm{H}\right)\right)
\nonumber\\&&+\rho\sum_{k=1}^K\sum_{t=0}^Tb_{k,t},
\end{eqnarray}
and
the $(k,t)$-th \textcolor{black}{element} of $\mathbf{d}^{(i-1)}$
\textcolor{black}{can be written as}
\begin{eqnarray}\label{dkt}
d^{(i-1)}_{k,t}
&=&\rho+\sum_{m=1}^M\Bigg(\mathbf{s}_{k,t}^\mathrm{H}\left(\mathbf{C}_m^{(i-1)}\right)^{-1}\mathbf{s}_{k,t}\nonumber\\
&&-\frac{1}{N}\mathbf{s}_{k,t}^\mathrm{H}\left(\mathbf{C}_m^{(i-1)}\right)^{-1}\mathbf{Y}_m\mathbf{Y}_m^\mathrm{H}\left(\mathbf{C}_m^{(i-1)}\right)^{-1}\mathbf{s}_{k,t}
\Bigg).\nonumber\\
\end{eqnarray}
While the update in \eqref{PG} still involves a nonsmooth and nonconvex problem,
the following proposition shows that its global optimal solution can be derived
in a closed form.

\newtheorem{proposition}{Proposition}
\begin{proposition}\label{update2}
The update in \eqref{PG} can be simplified in a closed form:
\begin{eqnarray}\label{closed}
b_{k,t}^{(i)}=
\begin{cases}
\textcolor{black}{\Pi_{[0,1]}}\left(\alpha_{k,t}^{(i)}+\eta_{i}\rho\right), &\text {if } t=\tau(k), \cr
\textcolor{black}{\Pi_{[0,1]}}\left(\alpha_{k,t}^{(i)}\right),  &\text {otherwise},
\end{cases}
\nonumber\\\forall k=1,\ldots,K,~~\forall t=0,\ldots,T,
\end{eqnarray}
where $\alpha_{k,t}^{(i)}\triangleq b_{k,t}^{(i-1)}-\eta_{i} d_{k,t}^{(i-1)}$,
$\tau(k)\in\arg\max\limits_{t\in\{0,\ldots,T\}}\alpha_{k,t}^{(i)}$,
\textcolor{black}{and $\Pi_{[0,1]}(\cdot)$ is the projection \textcolor{black}{onto} $[0,1]$.}
\end{proposition}
\begin{proof}
See Appendix \ref{proof of update2}.
\end{proof}

\vspace{2mm}

By iteratively updating $\mathbf{b}^{(i)}$ using $\eqref{closed}$,
the proposed algorithm for solving problem \eqref{opt2} is shown in Algorithm~1.
\textcolor{black}{While problem \eqref{opt2} is nonsmooth and nonconvex,
the} following theorem shows that Algorithm~1 is guaranteed to converge to
a stationary point.

\begin{algorithm}[t!]
\caption{Proposed \textcolor{black}{Centralized} Algorithm for Solving Problem \eqref{opt2}}
\begin{algorithmic}[1]\footnotesize
\State \textbf{Initialize} $\mathbf{b}^{(0)}$;\\
$\textbf{repeat}$ ($i=1,2,\ldots$)\\
~~~~Calculate the gradient $\mathbf{d}^{(i-1)}$ at
$\mathbf{b}^{(i-1)}$;\\
~~~~$\alpha_{k,t}^{(i)}=b_{k,t}^{(i-1)}-\eta_{i} d_{k,t}^{(i-1)},
\textcolor{black}{\forall k=1,\ldots,K}, \forall t=0,\ldots,T$;\\
~~~~Take any $\tau(k)\in\arg\max\limits_{t\in\{0,\ldots,T\}}\alpha_{k,t}^{(i)}, \textcolor{black}{\forall k=1,\ldots,K}$;\\
~~~~$
b_{k,t}^{(i)}=
\begin{cases}
\textcolor{black}{\Pi_{[0,1]}}\left(\alpha_{k,t}^{(i)}+\eta_{i} \rho\right), &\text {if } t=\tau(k) \cr
\textcolor{black}{\Pi_{[0,1]}}\left(\alpha_{k,t}^{(i)}\right),  &\text {otherwise}
\end{cases}, \textcolor{black}{\forall k=1,\ldots,K}$;\\
$\textbf{until}$ convergence
\end{algorithmic}
\end{algorithm}

\begin{theorem}\label{stationary}
When $\eta_i<1/L_{\textup{d}}$, with $L_{\textup{d}}$ denoting the Lipschitz constant of the gradient of \textcolor{black}{$G_0(\mathbf{b})$},
any limit point of the sequence generated by Algorithm~1 is
a stationary point of problem \eqref{opt2}.
\end{theorem}
\begin{proof}
See Appendix \ref{proof of stationary}.
\end{proof}

\vspace{2mm}

\textcolor{black}{Combining Theorem~1 and Theorem~2,
Algorithm~1 is not only guaranteed to converge to
a stationary point \textcolor{black}{problem} \eqref{opt2}, but also a stationary point of the original problem \eqref{opt1}.
Moreover, we}
can see from Algorithm~1 that $\mathbf{b}^{(i)}$
at each iteration is not enforced to satisfy the discontinuous nonconvex $\ell_0$-norm constraints in \eqref{opt1c2}, \textcolor{black}{but rather} these hard constraints are gradually satisfied to
reach a stationary point of problem \eqref{opt1}.
Therefore, the sequence \textcolor{black}{generated by the proposed algorithms} will \textcolor{black}{probably} not \textcolor{black}{get} stuck at bad stationary points
\textcolor{black}{caused} by the highly nonconvex $\ell_0$-norm. The performance gain over
the $\ell_0$-norm constrained BCD algorithm \cite{Wang22} will be shown through simulations in Section~VI.

\textcolor{black}{
The computational complexity of Algorithm~1 is
dominated by line~3, where the gradient is calculated
in \eqref{dkt}. Using the rank-1 update for the
matrix inverse in \eqref{dkt}, the computational complexity of line~3 is $\mathcal{O}\left(MK(T+1)(L+T)^2\right)$.
Thus, with $I$ denoting the iteration number, the overall
computational complexity of Algorithm~1 is $\mathcal{O}\left(IMK(T+1)(L+T)^2\right)$.}


\section{Distributed Algorithm for Asynchronous Activity Detection}
In this section, to reduce the computational \textcolor{black}{burden} at the CPU,
we further propose a distributed algorithm for
asynchronous activity detection.
Different from the centralized detection algorithm that is totally performed at the CPU,
the proposed distributed algorithm is executed \textcolor{black}{iteratively}
at both \textcolor{black}{the APs} and the CPU.
\textcolor{black}{At} each iteration,
each AP performs a local detection for the devices
and then sends the detection results to the CPU
for further processing.
\textcolor{black}{The combined result is then forwarded to the APs for the next iteration's computation.}

\textcolor{black}{First}, we transform problem \eqref{opt2} into an
equivalent consensus form:
\begin{subequations}\label{opt3}
\begin{eqnarray}\label{opt3cost}
\min_{\left\{\mathbf{x}_m\right\}_{m=1}^M, \mathbf{b}\in\left[0,1\right]^{K(T+1)}}
&&\rho\sum_{k=1}^K \left(\sum_{t=0}^Tb_{k,t}-\max_{t\in\{0,\ldots,T\}}b_{k,t}\right)\nonumber\\
&&+\sum_{m=1}^Mf_m\left(\mathbf{x}_{m}\right),
\end{eqnarray}
\begin{eqnarray}\label{opt3c1}
\text{s.t.}\
~~~~~~~~~~~&&\mathbf{x}_m=\mathbf{b},~~\forall m=1,\ldots,M,
\end{eqnarray}
\end{subequations}
where $f_m\left(\mathbf{x}_{m}\right)\triangleq\log\vert\tilde{\mathbf{C}}_m\vert+
\text{Tr}\left(\tilde{\mathbf{C}}_m^{-1}\mathbf{Y}_m\mathbf{Y}_m^\mathrm{H}\right)\textcolor{black}{/N}$,
and $\tilde{\mathbf{C}}_m$ is in the form of $\mathbf{C}_m$ but with
$\mathbf{x}_m$ replacing $\mathbf{b}$ in \eqref{covariance}.
\textcolor{black}{Notice that $f_m(\cdot)$ is a local function for the $m$-th AP and depends only on its own received signal $\mathbf{Y}_m$. This makes
problem \eqref{opt3} become a local problem when handling
$\mathbf{x}_{m}$ with a fixed $\mathbf{b}$.}

To solve problem \eqref{opt3} in a distributed manner, we
write its augmented Lagrangian function:
\begin{eqnarray}\label{ALF}
&&\mathcal{L}\left(\left\{\mathbf{x}_m\right\}_{m=1}^M, \mathbf{b};
\left\{\boldsymbol{\lambda}_m\right\}_{m=1}^M
\right)\nonumber\\
&=&\sum\limits_{m=1}^Mf_m\left(\mathbf{x}_{m}\right)
+\rho\sum\limits_{k=1}^K \left(\sum\limits_{t=0}^Tb_{k,t}-\max\limits_{t\in\{0,\ldots,T\}}b_{k,t}\right)
\nonumber\\
&&+\sum\limits_{m=1}^{M}\boldsymbol{\lambda}_m^\mathrm{T}\left(\mathbf{x}_m-\mathbf{b}\right)
+\frac{\mu}{2}\sum\limits_{m=1}^{M}\left\Vert\mathbf{x}_m-\mathbf{b}\right\Vert_2^2,
\end{eqnarray}
where each $\boldsymbol{\lambda}_m\in\mathbb{R}^{K(T+1)}$ is a dual variable corresponding to the
equality constraint $\mathbf{x}_m=\mathbf{b}$,
and $\mu>0$ is a penalty parameter to penalize the violation
of all the constraints in \eqref{opt3c1}.
\textcolor{black}{The appearance of \eqref{opt3} might suggest to use}
the classical ADMM algorithm \cite{Boyd2011},
which minimizes the augmented Lagrangian function \eqref{ALF} over
$\mathbf{b}$ and
$\left\{\mathbf{x}_m\right\}_{m=1}^M$ \textcolor{black}{alternatingly}.
However, since the term
$\rho\sum\limits_{k=1}^K \max\limits_{t\in\{0,\ldots,T\}}b_{k,t}$ in
\eqref{ALF} is nonsmooth and nonconvex, the classical ADMM algorithm cannot
guarantee its convergence.

To guarantee the convergence, we add an additional \textcolor{black}{proximal} term $\delta/2\left\Vert\mathbf{b}-\mathbf{b}^{(i-1)}\right\Vert_2^2$
to \textcolor{black}{\eqref{ALF}, making} the subproblem with respect to $\mathbf{b}$ at the $i$-th iteration \textcolor{black}{appear as}
\begin{eqnarray}\label{ADMM1}
&\min\limits_{\mathbf{b}\in\left[0,1\right]^{K(T+1)}}&
\rho\sum_{k=1}^K \left(\sum_{t=0}^Tb_{k,t}-\max_{t\in\{0,\ldots,T\}}b_{k,t}\right)
\nonumber\\&&+\sum_{m=1}^{M}\left(\boldsymbol{\lambda}_m^{(i-1)}\right)^\mathrm{T}\left(\mathbf{x}_m^{(i-1)}-\mathbf{b}\right)
\nonumber\\
&&+\frac{\mu}{2}\sum_{m=1}^{M}\left\Vert\mathbf{x}_m^{(i-1)}-\mathbf{b}\right\Vert_2^2+\frac{\delta}{2}\left\Vert\mathbf{b}-\mathbf{b}^{(i-1)}\right\Vert_2^2,
\nonumber\\
\end{eqnarray}
where $\delta>0$ is a parameter for \textcolor{black}{controlling the} convergence.
Although subproblem \eqref{ADMM1} is still nonsmooth and nonconvex,
the following proposition shows that its global optimal solution can be derived
in a closed form, \textcolor{black}{which can be proved using the same \textcolor{black}{argument as in} Appendix~\ref{proof of update2}.}

\begin{proposition}\label{ADMMclosed}
The closed-form solution of subproblem \eqref{ADMM1} is given by
\begin{eqnarray}\label{ADMMsolution1}
b_{k,t}^{(i)}=
\begin{cases}
\textcolor{black}{\Pi_{[0,1]}}\left(\beta_{k,t}^{(i)}+\frac{\rho}{\delta+M\mu}\right),~~&\text{if}~~t=\upsilon(k),
\cr
\textcolor{black}{\Pi_{[0,1]}}\left(\beta_{k,t}^{(i)}\right),~~& \text {otherwise},
\end{cases}\nonumber\\
\forall k=1,\ldots,K,~~\forall t=0,\ldots,T,
\end{eqnarray}
where $\beta_{k,t}^{(i)}\triangleq\left(\delta b_{k,t}^{(i-1)}+\sum_{m=1}^M\left(\mu x_{m,k,t}^{(i-1)}+\lambda_{m,k,t}^{(i-1)}\right)-\rho\right)$\\$/\left(\delta+M\mu\right)$ and
$\upsilon(k)\in\arg\max\limits_{t\in\{0,\ldots,T\}}\beta_{k,t}^{(i)}$.
\end{proposition}

On the other hand, the subproblem with respect to $\left\{\mathbf{x}_m\right\}_{m=1}^M$
at the $i$-th iteration
can be decomposed into $M$ parallel subproblems, with each written as
\begin{equation}\label{ADMM2}
\min _{\mathbf{x}_m}~~
f_m\left(\mathbf{x}_{m}\right)
+\left(\boldsymbol{\lambda}_m^{(i-1)}\right)^\mathrm{T}
\left(\mathbf{x}_m-\mathbf{b}^{(i)}\right)
+\frac{\mu}{2}\left\Vert\mathbf{x}_m-\mathbf{b}^{(i)}\right\Vert_2^2.
\end{equation}
\textcolor{black}{Compared}
to the single-cell device activity detection problem \cite{Haghighatshoar2018,chenIcc,icasspW},
\textcolor{black}{subproblem \eqref{ADMM2} only differs in the}
additional linear and quadratic terms.
\textcolor{black}{Therefore,}
subproblem \eqref{ADMM2} can be solved
\textcolor{black}{in a similar way to \cite{Haghighatshoar2018,chenIcc,icasspW}}
by updating each coordinate of $\mathbf{x}_m$
sequentially with the coordinate descent algorithm.
In particular, with other coordinates fixed,
$x_{m,k,t}$ is updated by solving
\begin{eqnarray}\label{cd}
\min _{x_{m,k,t}}~~&&\log\left(1 + \xi_{1}x_{m,k,t}\right)-\frac{\xi_{2}x_{m,k,t}}{1+\xi_{1}x_{m,k,t}}
\nonumber\\
&&+\lambda_{m,k,t}^{(i-1)}
\left(x_{m,k,t}-b_{k,t}^{(i)}\right)+\frac{\mu}{2}\left(x_{m,k,t}-b_{k,t}^{(i)}\right)^2,\nonumber\\
\end{eqnarray}
where $\xi_{1} \triangleq p_kg_{k,m}\mathbf{s}_{k,t}^{\mathrm{H}}\mathbf{D}_{m,k,t}^{-1}\mathbf{s}_{k,t}$,
$\xi_{2} \triangleq p_kg_{k,m}/{N}\mathbf{s}_{k,t}^{\mathrm{H}} \mathbf{D}_{m,k,t}^{-1}\mathbf{Y}_m\mathbf{Y}_m^\mathrm{H}  \mathbf{D}_{m,k,t}^{-1} \mathbf{s}_{k,t}$,
and
 $\mathbf{D}_{m,k,t} \triangleq \sum\limits_{(\bar{k}, \bar{t}) \neq (k,t)}$
 $x_{m,\bar{k},\bar{t}}p_{\bar{k}}g_{\bar{k},m}\mathbf{s}_{\bar{k},\bar{t}}\mathbf{s}_{\bar{k},\bar{t}}^{\mathrm{H}} + \sigma_{m}^{2}\mathbf{I}_{L+T}$.
Notice that $\xi_{1}$ and $\xi_{2}$
depend on $\left(m,k,t\right)$.
For notational simplicity, the $\left(m,k,t\right)$ dependence is not explicitly stated.
Setting the gradient of \eqref{cd} to zero yields
\begin{eqnarray}\label{grad}
&&\left(1 + \xi_{1}x_{m,k,t}\right)\xi_{1}-\xi_{2}+\lambda_{m,k,t}^{(i-1)}\left(1 + \xi_{1}x_{m,k,t}\right)^2\nonumber\\
&&+\mu\left(x_{m,k,t}-b_{k,t}^{(i)}\right)\left(1 + \xi_{1}x_{m,k,t}\right)^2=0,
\end{eqnarray}
whose roots can be expressed using the cubic formula. Consequently,
the optimal solution of \eqref{cd} is obtained by selecting the root
with the minimum cost function value.

After updating $\mathbf{b}$ and $\left\{\mathbf{x}_m\right\}_{m=1}^M$,
the dual variables are updated by a dual ascent step:
\begin{equation}\label{dual}
\boldsymbol{\lambda}_m^{(i)}=
\boldsymbol{\lambda}_m^{(i-1)}+\mu\left(\mathbf{x}_m^{(i)}-\mathbf{b}^{(i)}\right),~~
\forall m=1,\ldots,M.
\end{equation}
By \textcolor{black}{iteratively} updating the primal and dual variables,
the proposed distributed algorithm for solving problem \eqref{opt3} is summarized in Algorithm~2.
The following theorem shows that Algorithm~2 is guaranteed to converge to
a stationary point of problem \eqref{opt3}.

\begin{algorithm}[t!]
\caption{Proposed \textcolor{black}{Distributed} Algorithm for Solving Problem \eqref{opt3}}
\begin{algorithmic}[1]\footnotesize
\State \textbf{Initialize} $\mathbf{b}^{(0)}$, $\left\{\mathbf{x}_m^{(0)}\right\}_{m=1}^M$, and $\left\{\boldsymbol{\lambda}_m^{(0)}\right\}_{m=1}^M$;\\
$\textbf{repeat}$ ($i=1,2,\ldots$)\\
~~~~\textcolor{black}{Each AP $m$ sends $\mu\mathbf{x}_m^{(i-1)}+\boldsymbol{\lambda}_m^{(i-1)}$ to the CPU, $\forall m=1,\ldots,M$};\\
~~~~The CPU updates $\mathbf{b}^{(i)}$ with \eqref{ADMMsolution1};\\
~~~~\textcolor{black}{The CPU broadcasts $\mathbf{b}^{(i)}$ to each AP $m$, $\forall m=1,\ldots,M$};\\
~~~~\textcolor{black}{Each AP $m$} updates $\mathbf{x}_m^{(i)}$ \textcolor{black}{by solving problem \eqref{ADMM2} with}
the coordinate descent algorithm, $\forall m=1,\ldots,M$;\\
~~~~\textcolor{black}{Each AP $m$} updates $\boldsymbol{\lambda}_m^{(i)}$ with \eqref{dual}, $\forall m=1,\ldots,M$;\\
$\textbf{until}$ convergence
\end{algorithmic}
\end{algorithm}

\begin{theorem}\label{stationary2}
When $\mu>2L_{m}$, with $L_{m}$ denoting the Lipschitz constant of $\nabla f_m\left(\mathbf{x}_{m}\right)$,
any limit point of the sequence generated by Algorithm~2 is a stationary point of problem \eqref{opt3}.
\end{theorem}
\begin{proof}
See Appendix \ref{proof of stationary2}.
\end{proof}

\vspace{2mm}

Notice that problem \eqref{opt3} is an equivalent reformulation of problem \eqref{opt2}, which is also equivalent to the original problem \eqref{opt1} (see Theorem~1).
Therefore, we can conclude that Algorithm~2 is also guaranteed to converge to a stationary point of problem \eqref{opt1}.

\textcolor{black}{The computational complexity of Algorithm~2 is
dominated by line~6, where the coordinate descent algorithm is applied to solve \eqref{ADMM2}
with computational complexity $\mathcal{O}\left(K(T+1)(L+T)^2\right)$ at each AP \cite{chenIcc}. Due to the multi-AP parallel computation, with $I$ denoting the iteration number,
the time complexity of Algorithm~2 is $\mathcal{O}\left(IK(T+1)(L+T)^2\right)$.}

\section{\textcolor{black}{Communication-Efficient Enhancement for Algorithm~2}}
In \textcolor{black}{Algorithm~1 and Algorithm~2},
we assume that the received signals \textcolor{black}{or the local detection results}
can be accurately collected at the CPU.
In practice, the APs and the CPU
are connected by capacity-limited fronthaul links.
Therefore, the \textcolor{black}{exchanged contents}
have to be compressed before being \textcolor{black}{transmitted.
In Algorithm~1, the received signals forwarded to the CPU}
are in a large dynamic range due to the randomness of the activities, delays, and channels.
\textcolor{black}{Thus,} the \textcolor{black}{required} number of bits for compression
\textcolor{black}{could} be very large in order to maintain \textcolor{black}{a} high
fidelity. If the compression error is large, it inevitably degrades the detection performance.
\textcolor{black}{On the other hand, in Algorithm~2},
we can see that instead of sending the received signals,
each AP only requires to forward the local detection results to the CPU.
However, as a distributed algorithm, Algorithm~2 requires communications
between the APs and the CPU at each iteration.
Therefore, the communication overheads of Algorithm~2
are affected by its required iteration number for convergence.
This brings a difficult dilemma of reducing the capacity requirements of the fronthaul links while still achieving a satisfactory detection performance.

In \textcolor{black}{this} section, we resolve this dilemma by \textcolor{black}{
modifying Algorithm~2 such that the convergence is accelerated and
the exchanged variables appear in a much smaller dynamic range,
which reduces the required
number of bits for compression} without sacrificing \textcolor{black}{the} detection performance.
In particular, we observe that at each iteration of Algorithm~2,
the CPU updates $\mathbf{b}$ by solving the subproblem \eqref{ADMM1},
where the variable $\mathbf{b}$ is optimized to minimize the Euclidean distance
to the local detection result $\mathbf{x}_m^{(i-1)}$ at each AP.
However, from the original problem \eqref{opt1}, we can see that
$\mathbf{b}$ is actually optimized through the covariance matrix $\mathbf{C}_m$ by
minimizing a distance \textcolor{black}{defined using}
\textcolor{black}{the} likelihood function.
Based on this observation, we \textcolor{black}{replace the terms with respect to \textcolor{black}{the}
Euclidean distance (i.e., third and fourth terms)
in subproblem \eqref{ADMM1} with
the form of \eqref{opt1} and drop the term with respect to \textcolor{black}{the} dual variables
to further reduce the communication overheads}.
Consequently, subproblem \eqref{ADMM1} is
modified as
\begin{eqnarray}\label{modify}
\min _{\mathbf{b}\in\left[0,1\right]^{K(T+1)}}&&
\sum_{m=1}^M\left(\log\vert\mathbf{C}_m\vert+\text{Tr}\left(\mathbf{C}_m^{-1}\tilde{\mathbf{C}}_m^{(i-1)}\right)\right)
\nonumber\\&&+\rho\sum_{k=1}^K \left(\sum_{t=0}^Tb_{k,t}-\max_{t\in\{0,\ldots,T\}}b_{k,t}\right),
\end{eqnarray}
where $\tilde{\mathbf{C}}_m^{(i-1)}$ is in the form of $\mathbf{C}_m$ but with
$\mathbf{x}_m^{(i-1)}$ replacing $\mathbf{b}$ in \eqref{covariance}.
Since problem \eqref{modify} is \textcolor{black}{in the} same form of problem \textcolor{black}{\eqref{opt2}},
we can adopt a similar algorithm to Algorithm~1 to solve it
\textcolor{black}{(simply replace $\mathbf{Y}_m\mathbf{Y}_m^\mathrm{H}/N$ with $\tilde{\mathbf{C}}_m^{(i-1)}$ in Algorithm 1)}.
After replacing line~3 by the above updating procedure,
Algorithm~2 requires much fewer iterations, which will be verified via simulations.

We can interpret the benefit of the above modification from \eqref{ADMM1} to \eqref{modify} as follows. In problem \eqref{ADMM1}, $\mathbf{b}$ is estimated by
minimizing the Euclidean distance to each $\mathbf{x}_m^{(i-1)}$.
Nevertheless, \textcolor{black}{due to the diverse distances of each device from different APs,
the detection accuracy of each device} from different APs
can be substantially different. In particular, the detection results of a device
obtained from its nearby APs are more reliable than those from \textcolor{black}{distant} APs.
However, the information on \textcolor{black}{this} detection reliability is not
\textcolor{black}{captured and}
modeled in problem \eqref{ADMM1}.
In contrast, problem \eqref{modify} adopts $\tilde{\mathbf{C}}_m^{(i-1)}$
based on $\mathbf{x}_m^{(i-1)}$ as well as the corresponding large-scale fading components, which
\textcolor{black}{successfully capture and use this} detection reliability \textcolor{black}{information}.
In this sense, problem \eqref{modify} can provide a better
formulation for the approximation of $\mathbf{b}$,
making \textcolor{black}{the convergence of the iterative algorithm}
much faster.

\begin{algorithm}[t!]
\caption{Communication-Efficient Enhancement for Algorithm~2}
\begin{algorithmic}[1]\footnotesize
\State \textbf{Initialize} $\mathbf{b}^{(0)}$, $\left\{\mathbf{x}_m^{(0)}\right\}_{m=1}^M$, and $\left\{\boldsymbol{\lambda}_m^{(0)}\right\}_{m=1}^M$;\\
$\textbf{repeat}$ ($i=1,2,\ldots$)\\
~~~~Each AP $m$ sends $\mathcal{Q}\left(\mathbf{x}_m^{(i-1)}\right)$ to the CPU, $\forall m=1,\ldots,M$;\\
~~~~The CPU updates
$\mathbf{b}^{(i)}$ by solving \eqref{modify} with a similar algorithm to Algorithm~1;\\
~~~~The CPU broadcasts $\mathcal{Q}\left(\mathbf{b}^{(i)}\right)$ to each AP \textcolor{black}{$m$, $\forall m=1,\ldots,M$};\\
~~~~\textcolor{black}{Each AP $m$} updates $\mathbf{x}_m^{(i)}$ \textcolor{black}{by solving problem \eqref{ADMM2} with}
the coordinate descent algorithm, $\forall m=1,\ldots,M$;\\
~~~~\textcolor{black}{Each AP $m$} updates $\boldsymbol{\lambda}_m^{(i)}$ with \eqref{dual}, $\forall m=1,\ldots,M$;\\
$\textbf{until}$ convergence
\end{algorithmic}
\end{algorithm}

Using the modification in \eqref{modify},
we summarize the resulting algorithm as Algorithm~3.
In line~3 and line~5, $\mathcal{Q}(\cdot)$ is a function to compress
$\mathbf{x}_m$ and $\mathbf{b}$. For instance, \textcolor{black}{a simple $\mathcal{Q}(\cdot)$ that
we can adopt is the
uniform scalar quantizer for each element of the vectors}.
Due to the small dynamic range $[0,1]$ of $\mathbf{x}_m$ and $\mathbf{b}$,
the required number of quantization bits for compression can be
much smaller than that in the centralized detection.
\textcolor{black}{Furthermore}, since most of the entries in $\mathbf{x}_m$ and $\mathbf{b}$ are zeros due to the sparse activities, we can also adopt \textcolor{black}{the} variable-length compression scheme
such as Huffman coding to further reduce the required number of bits, making
Algorithm~3 more communication-efficient.
The enhancement on \textcolor{black}{the} communication efficiency of Algorithm~3
will be shown through simulations in Section~VI.
\textcolor{black}{On the other hand, similar to Algorithm~2, the computational complexity of Algorithm~3 is
also dominated by line~6 with $\mathcal{O}\left(K(T+1)(L+T)^2\right)$,
and its time complexity is $\mathcal{O}\left(IK(T+1)(L+T)^2\right)$,
where $I$ denotes the iteration number.}

\section{Simulation Results}
In this section, we present the performance of the proposed centralized and distributed
algorithms via simulations
in terms of the probability of missed detection (PM), i.e.,
the probability that an active device is detected as
inactive or its delay is incorrectly detected, and
the probability of false alarm (PF), i.e., the probability that an inactive device is detected as
active \cite{Wang22}.
Specifically, the indicator of the device activity and delay is recovered by
$\hat{b}_{k,t}=\mathbb{I}(b_{k,t}>\gamma)$, where $b_{k,t}$
is the optimization result \textcolor{black}{returned} by the proposed algorithms and $\gamma$ is a threshold that \textcolor{black}{varies} in $[0, 1]$ to realize a trade-off between PM and PF.

\subsection{Simulation Setting}
We consider a $1 \times 1$ square kilometers area with wrap-around at the boundary.
There are $M$ APs and $K=100$ IoT devices uniformly distributed in this square area, where
the ratio of the active devices to the total devices is $0.1$.
The signature sequence of each device is an independently
generated complex Gaussian distributed vector with i.i.d. elements and each element is with zero
mean and unit variance. The large-scale fading component follows the
micro-cell propagation model \cite{3GPP17}, i.e.,
$g_{k,m}=-30.5-36.7 \log_{10}D_{k,m}+\Psi_{k,m}$ in dB,
where $D_{k,m}$ is the distance in meters between the $k$-th device and the $m$-th AP.
\textcolor{black}{To reflect the effect of blockage,
the large-scale fading component
includes a shadow fading component
$\Psi_{k,m}$, which is complex Gaussian distributed with mean $0$ and variance $4$~\cite{3GPP17}.}
The maximum transmit power of each device is $23$ dBm and
the background Gaussian noise power is $-104$ dBm.
In order to reduce the channel gain differences among different devices,
the transmit power of each device is controlled based on the large-scale fading components
such that the SNR at its \textcolor{black}{dominant AP (which is the AP with the largest channel gain)} is fixed to a target value that
can be achieved by $95\%$ of the active devices \cite{Ganesan2021}.
All the simulation results are obtained by averaging over $1000$ trials,
with independent APs' and devices' locations, channels,
signature sequences, device activity patterns, delays,
and noise realizations in each trial.

In Algorithm~1, we set the penalty parameter $\rho$ as $0.16$,
and choose an adaptive step size $\eta_i$
as the inverse of the local estimation of the Lipschitz constant of the gradient \cite{Shao2020-2}.
Moreover, $\mathbf{b}^{(0)}$ is initialized as a zero vector.
In Algorithm~2 and Algorithm~3, we set $\mu$ as $0.08$,
and initialize $\mathbf{b}^{(0)}$ and $\boldsymbol{\lambda}_m^{(0)}$ as zero vectors.
To achieve fast convergence, $\mathbf{x}_m^{(0)}$ is initialized
by a local detection at each AP:
$\min\limits_{\mathbf{x}_m\in[0,1]^{K(T+1)}} f_m\left(\mathbf{x}_{m}\right)$,
which can be solved with the coordinate descent algorithm for single-cell device activity detection
\cite{Haghighatshoar2018,chenIcc,icasspW}.

\subsection{Proposed Centralized Algorithm \textcolor{black}{Versus State-of-the-Art Approaches}}
First, we demonstrate the performance of the proposed centralized Algorithm~1.
For comparison, we also show the detection performance
of two benchmarks, i.e., CD-E and BCD in \textcolor{black}{\cite{chenIcc,Wang22}}. While these two approaches
are designed for single-cell asynchronous activity detection, we extend them to
solve problem \eqref{opt1} for cell-free massive MIMO as follows.
\begin{itemize}
     \item CD-E first drops the $\ell_0$-norm constraints in \eqref{opt1c2} and then solves the relaxed problem with the \textcolor{black}{coordinate} descent algorithm. After the optimization process, \eqref{opt1c2} is re-enforced by an additional constraint enforcement step.

     \item BCD enforces \eqref{opt1c2} within the optimization process. Specifically, the variable $\mathbf{b}$ is decomposed into $K$ blocks, where each block $\mathbf{b}_k$ is sequentially updated with other blocks fixed. Each subproblem is
         solved by comparing the solutions of $T+1$ one-dimensional subproblems
         within the feasible set of \eqref{opt1c2}
         and then selecting the one with the minimum cost function value.

\end{itemize}

We compare the performance of different approaches in terms of PM and PF in Fig. \ref{F1},
where the numbers of APs and antennas at each AP, the length of the signature sequences,
and the maximum delay are set as $M=N=8$,
$L=9$, and $T=1$, respectively. We can see that
both the proposed Algorithm~1 and BCD outperform CD-E,
since they both solve the original problem \eqref{opt1} to stationary points.
\textcolor{black}{However}, Algorithm~1 achieves a much better PM-PF trade-off than that of
BCD. In particular, the PM of Algorithm~1 is over $10$ times lower than that of BCD
under the same PF.
This is because in Algorithm~1,
the highly nonconvex $\ell_0$-norm constraints in \eqref{opt1c2}
are gradually satisfied in a gentle fashion as the iteration number increases,
which \textcolor{black}{is helpful in getting}
around bad stationary points of problem \eqref{opt1}.

\begin{figure}[t!]
\begin{center}
  \includegraphics[width=0.45\textwidth]{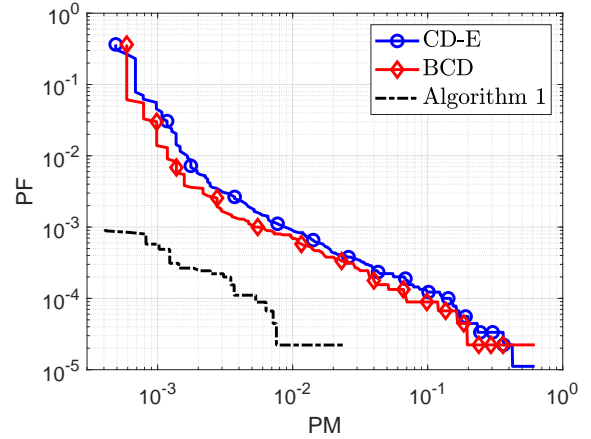}
  \caption{PM-PF trade-offs achieved by different centralized algorithms.}\label{F1}
\end{center}
\vspace{-0.2cm}
\end{figure}

We further compare these centralized algorithms under
different maximum delays and different numbers of APs.
Due to the trade-off between PM and PF, we show the probability of error when PF $=$ PM by appropriately setting the threshold $\gamma$. The probability of error versus the maximum delay $T$ is shown in Fig.~\ref{F2a},
where the numbers of APs and antennas at each AP are set as $M=N=8$ and
the length of the effective signature sequences is fixed as $L+T=10$.
In particular, $(L=10,~T=0)$ represents ideal synchronous transmissions with no delay.
It can be seen that while all the approaches result in higher PM as $T$ increases \textcolor{black}{up to $8$, Algorithm~1 always performs the best, which demonstrates the superiority of Algorithm~1
under asynchronous transmissions.}

On the other hand, the probability of error versus the number of APs $M$ is shown in Fig.~\ref{F2b},
where the length of the signature sequences and the maximum delay are set as
$L=9$ and $T=1$, respectively. The total number of antennas at all the APs
is fixed as $MN=64$.
It can be seen that as $M$ increases from $1$ to $16$, the detection performance of all the approaches
becomes better. This performance improvement \textcolor{black}{comes from the fact that}
\textcolor{black}{when there are more APs in the area, the distance
between each device and each AP becomes shorter, resulting in
a higher SNR at each AP.}
On the other hand, when $M$ further increases, the antenna number at each AP $N$ becomes much smaller.
The significant decrease in the spatial resolution makes the detection performance worse.
Nevertheless, we can see that Algorithm~1 always performs better than \textcolor{black}{the other two} approaches in the whole \textcolor{black}{range of} $M$.

\textcolor{black}{To demonstrate the superiority of massive MIMO, we further show
the performance comparison under different numbers of antennas in Fig.~\ref{F2c}.
We fix the number of APs as $M=8$ and vary
the number of antennas at each AP such that the total number of antennas increases from $32$ to $96$.
It can be seen that while the detection performance of all the approaches
becomes better as the number of antennas increases,
the proposed Algorithm~1 always achieves the best performance.}
Due to the superiority of Algorithm~1, we adopt it as a baseline
for the proposed distributed algorithms in the rest of simulations.

\begin{figure}[t!]
\begin{center}
 \subfigure[Probability of error versus $T$.]{
  \includegraphics[width=0.45\textwidth]{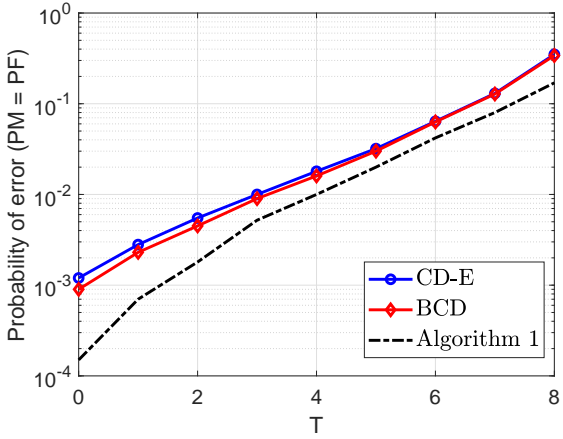}\label{F2a}
  }
  \subfigure[Probability of error versus $M$.]{
  \includegraphics[width=0.45\textwidth]{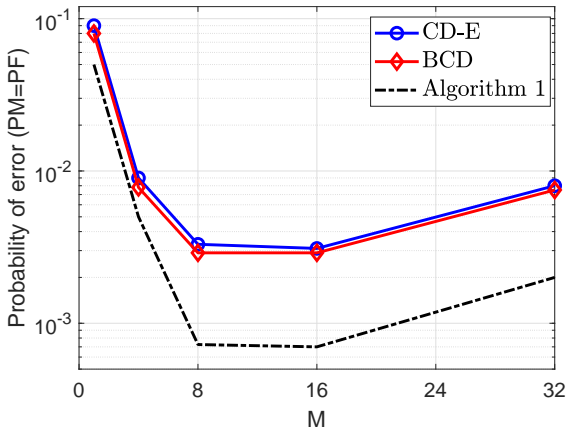}\label{F2b}
}
  \subfigure[Probability of error versus total number of antennas.]{
  \includegraphics[width=0.45\textwidth]{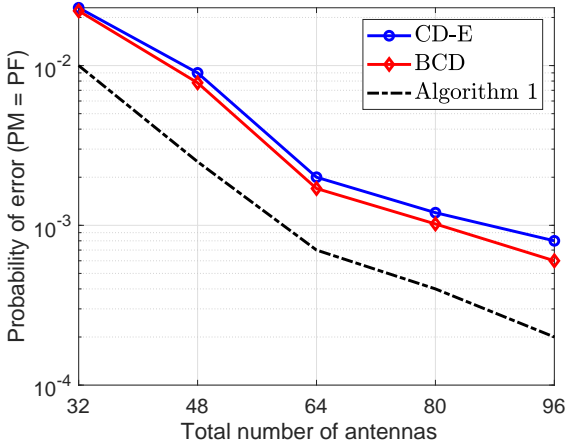}\label{F2c}
}
\caption{Performance comparison of centralized algorithms.}
\end{center}
\vspace{-0.2cm}
\end{figure}

\subsection{Proposed Distributed Algorithms \textcolor{black}{Versus Centralized Algorithm}}
\begin{figure}[t!]
\begin{center}
 \subfigure[Probability of error versus \textcolor{black}{iteration number} $I$.]{
  \includegraphics[width=0.45\textwidth]{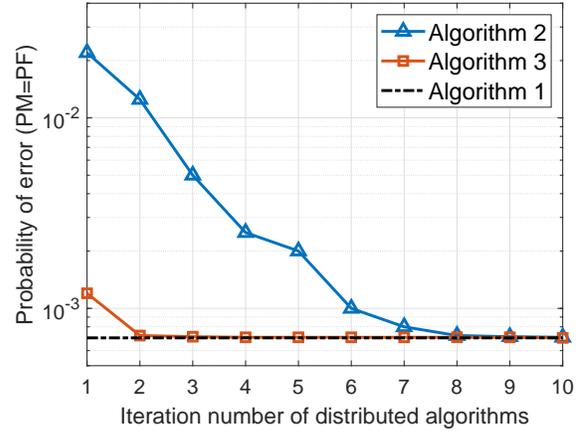}\label{F3a}
  }
  \subfigure[Probability of error versus $T$.]{
  \includegraphics[width=0.45\textwidth]{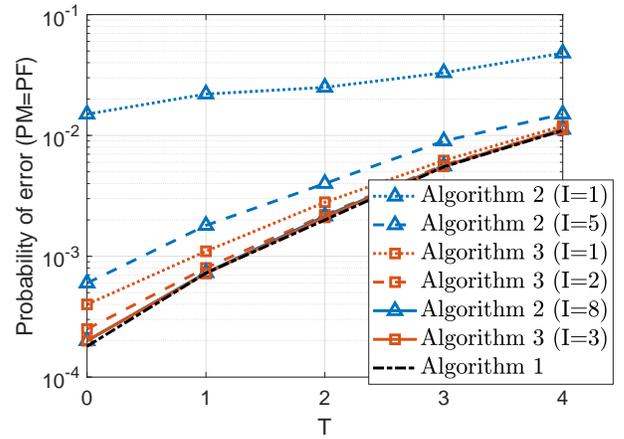}\label{F3b}
}\caption{Performance comparison of distributed algorithms.}
\end{center}
\vspace{-0.2cm}
\end{figure}

Next, we show the detection performance of the proposed distributed Algorithm~2 and Algorithm~3,
\textcolor{black}{and compare them to that of the centralized algorithm}.
The probability of error versus the iteration number of the distributed algorithms
is shown in Fig. \ref{F3a}, where $M=N=8$,
$L=9$, and $T=1$. For fair comparison, all these algorithms are executed under ideal fronthauls.
It can be seen that both Algorithm~2 and Algorithm~3 achieve fast convergence (within $8$ iterations)
to the result of Algorithm~1.
Furthermore, Fig. \ref{F3b} shows that \textcolor{black}{pretty} fast convergence can be achieved under different $T$. This means that Algorithm~2 or
Algorithm~3 can replace Algorithm~1 without changing the performance.
\textcolor{black}{On the other hand,}
since Algorithm~3 is judiciously designed based on
Algorithm~2, we can see that Algorithm~3 achieves an impressively \textcolor{black}{fast} convergence within only $2$ \textcolor{black}{to $3$} iterations.
\textcolor{black}{As} Algorithm~3 achieves the same performance to that of Algorithm~2 \textcolor{black}{but}
with much faster convergence, we only show the performance of Algorithm~3 under
capacity-limited fronthauls in the \textcolor{black}{following} simulations.

\subsection{\textcolor{black}{Performance Under Capacity-Limited Fronthauls}}

Next, we show how many \textcolor{black}{quantization} bits
are needed
to approach the results under ideal fronthauls.
For \textcolor{black}{the} illustration purpose, we adopt a uniform scalar quantizer
in Algorithm~1 and Algorithm~3, respectively.
The probability of error versus $T$
is shown in Fig.~\ref{F4a}, where $M=N=8$ and $L+T=10$.
It can be seen that under different $T$, Algorithm~3 always performs very close to Algorithm~1.
However, Algorithm~1 requires at least $16$ quantization bits per real-valued scalar
to approach the performance under ideal fronthauls.
\textcolor{black}{When $14$ quantization bits are used,
Algorithm~1 performs even worse than
Algorithm~3 with only $4$
quantization bits and $I=1$ iteration.}
In Fig.~\ref{F4b},
we show the probability of error versus $M$, where $L=9$, $T=1$, and $MN=64$.
Similar to Fig.~\ref{F4a},
we can see that under different $M$, Algorithm~3 always requires much fewer
quantization bits per real-valued scalar to approach the detection performance under ideal fronthauls.

\textcolor{black}{To clearly see the overall communication overheads, we analyze the number of bits required to be transmitted in \textcolor{black}{Algorithm}~1 and Algorithm~3 as follows.}
In Algorithm~1,
the received signal $\mathbf{Y}_m$ is used via $\mathbf{Y}_m\mathbf{Y}_m^\mathrm{H}/N$, which is a Hermitian matrix with $(L+T)^2$ real-valued scalars.
Thus, when $L+T\leq2N$, each AP $m$ can \textcolor{black}{quantize} and send $\mathbf{Y}_m\mathbf{Y}_m^\mathrm{H}/N$
instead of $\mathbf{Y}_m$ (with $2(L+T)N$ real-valued scalars) to the CPU.
With $Q_1$ denoting the number of quantization bits per real-valued scalar,
the overall number of bits required by Algorithm~1 is $MQ_1(L+T)^2$
when $L+T\leq2N$ or $2MQ_1(L+T)N$ otherwise.
On the other hand, in Algorithm~3, each AP $m$ sends
its \textcolor{black}{quantized} local detection result to the CPU at \textcolor{black}{each} iteration.
With 
$Q_2$ denoting the number of quantization bits per real-valued scalar,
the number of bits sent from each AP $m$ to the CPU is
\textcolor{black}{$Q_2K(T+1)$}. Similarly, after the CPU updates $\mathbf{b}^{(i)}$,
the number of bits sent from the CPU to each AP $m$ is also
\textcolor{black}{$Q_{2}K(T+1)$}. Since Algorithm~3 can stop before sending $\mathbf{b}^{(I)}$
to the APs at the last iteration, the overall number of bits required by Algorithm~3 is
\textcolor{black}{$(2I-1)MKQ_{2}(T+1)$}. For example,
when \textcolor{black}{$M=N=8$}, $L=9$, $T=1$, $I=1$, \textcolor{black}{$Q_1=14$}, and $Q_2=4$,
Algorithm~1 and Algorithm~3 require \textcolor{black}{$11200$} and $6400$ quantization bits, respectively. \textcolor{black}{This shows that even under the
uniform scalar quantization, Algorithm~3 reduces the number of bits transmitted significantly.}

\begin{figure}[t!]
\begin{center}
 \subfigure[Probability of error versus $T$.]{
  \includegraphics[width=0.45\textwidth]{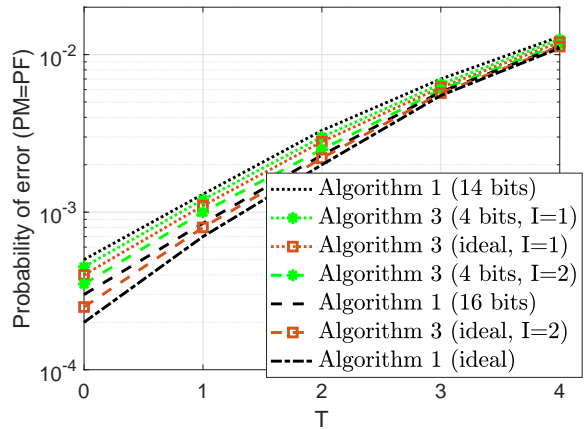}\label{F4a}
  }
  \subfigure[Probability of error versus $M$.]{
  \includegraphics[width=0.45\textwidth]{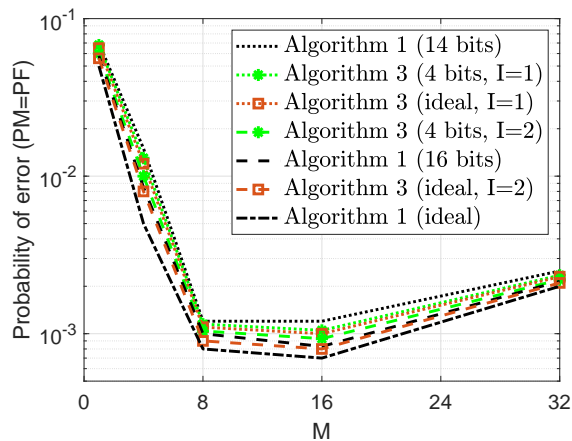}\label{F4b}
}\caption{Performance comparison between Algorithm~1 and Algorithm~3
\textcolor{black}{using the simple uniform scalar quantization}.}\label{F4}
\end{center}
\vspace{-0.2cm}
\end{figure}

\textcolor{black}{Since most of the local detection results are zero, the actual number of bits to be transmitted
can be further reduced by using various data compression schemes. As a demonstration,}
we compare the total number of bits required by Algorithm~1 and Algorithm~3
using Huffman coding.
In particular, Huffman coding is applied to the quantized contents for both Algorithm~1 and
Algorithm~3.
Figure~\ref{F5} shows the simulation results for
$M=N=8$,
$L=9$, and $T=1$.
\textcolor{black}{For Algorithm~1, we consider different quantization levels $2^{11}$, $2^{14}$, and $2^{16}$,
whereas for Algorithm 3, the quantization level is fixed as $2^4$ and each AP transmits
the local detection results of $50$, $70$, $80$, and $100$ devices with the largest large-scale fading coefficients.}
We can see that
while Huffman coding is effective in reducing the overall number of bits for both Algorithm~1 and Algorithm~3, the compression ratio is higher in Algorithm~3 due to the sparse
local detection results. For example, before using Huffman coding,
the number of bits required by Algorithm~3 is \textcolor{black}{almost} $2$ times smaller than that of
Algorithm~1, whereas after using Huffman coding,
the number of bits required by Algorithm~3 is at least $3$ times smaller than that of
Algorithm~1.

\begin{figure}[t!]
\begin{center}
  \includegraphics[width=0.45\textwidth]{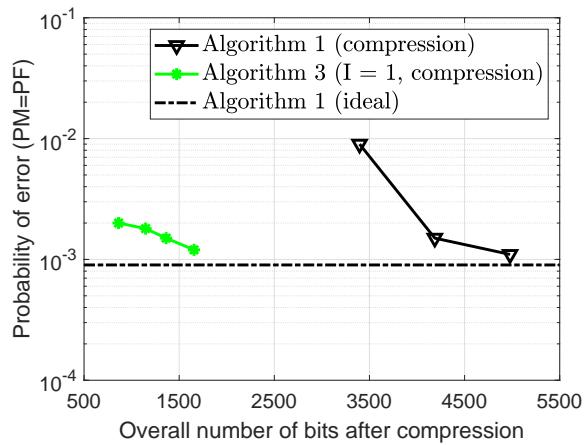}
  \caption{Probability of error versus overall number of bits \textcolor{black}{after Huffman coding} compression.}\label{F5}
\end{center}
\vspace{-0.2cm}
\end{figure}

\section{Conclusions}
This paper studied asynchronous activity detection methods for cell-free massive MIMO.
To tackle the discontinuous and nonconvex $\ell_0$-norm constraints due to the asynchronous transmissions,
an equivalent reformulation of the original problem \textcolor{black}{was established.
A centralized algorithm and a distributed algorithm were proposed,
with both being theoretically guaranteed to converge to a
stationary point of the original asynchronous activity detection problem}.
Since the proposed algorithms address the $\ell_0$-norm constraints in a gentle fashion
as the iteration number increases, the sequence \textcolor{black}{generated by the proposed algorithms} can get around bad stationary points
caused by the highly nonconvex $\ell_0$-norm.
To reduce the capacity requirements of the fronthauls in cell-free massive MIMO,
\textcolor{black}{a communication-efficient accelerated distributed algorithm
was further designed}.
Simulation results demonstrated that the
proposed centralized \textcolor{black}{and distributed algorithms outperform}
state-of-the-art approaches,
whereas the proposed \textcolor{black}{accelerated} distributed algorithm achieves a close detection performance to that of the centralized algorithm \textcolor{black}{but with a much smaller number of bits to be transmitted on the fronthaul links}.

\numberwithin{equation}{section}
\appendices

\section{Proof of Theorem \ref{equivalence}}\label{proof of equivalence}
We first prove that there exists a finite $\rho^*<\infty$ such that for any $\rho>\rho^*$,
the \textcolor{black}{stationary point} of problem \eqref{opt2}
is \textcolor{black}{also} a
feasible point of problem \eqref{opt1}.
\textcolor{black}{Let $\mathbf{b}^*$ denote a stationary point of problem \eqref{opt2},
which} satisfies the following first-order optimality condition\textcolor{black}{\cite{facchinei2007finite}}:
\begin{eqnarray}\label{first-order}
&&-\rho\sum_{k=1}^K
\lim _{\tau \rightarrow 0^{+}} \frac{\max\limits_{t\in\{0,\ldots,T\}}\left(b_{k,t}^*+\tau \textcolor{black}{\tilde{d}}_{k,t}\right)
- \max\limits_{t\in\{0,\ldots,T\}}b_{k,t}^*}{\tau}
\nonumber\\&&+\sum_{m=1}^M \nabla f_m\left(\mathbf{b}^*\right)^\mathrm{T}\textcolor{black}{\tilde{\mathbf{d}}}
+\rho\sum_{k=1}^K \sum_{t=0}^T\textcolor{black}{\tilde{d}}_{k,t}
\geq 0,\nonumber\\
&&\forall \textcolor{black}{\tilde{\mathbf{d}}}\in\mathcal{T}\left(\mathbf{b}^*\right),
\end{eqnarray}
where $f_m\left(\mathbf{b}\right)\triangleq\log\vert{\mathbf{C}}_m\vert+
\text{Tr}\left({\mathbf{C}}_m^{-1}\mathbf{Y}_m\mathbf{Y}_m^\mathrm{H}\right)\textcolor{black}{/N}$ and
$\mathcal{T}\left(\mathbf{b}^*\right)$ is
the tangent cone of the feasible set of problem \eqref{opt2} at $\mathbf{b}^*$.
\textcolor{black}{Denoting $\mathcal{W}(k)
\triangleq\arg\max\limits_{t\in\{0,\ldots,T\}} b_{k,t}^*$,
we have $\max\limits_{t\in\{0,\ldots,T\}}\left(b_{k,t}^*+\tau \textcolor{black}{\tilde{d}}_{k,t}\right)=\max\limits_{t\in\{0,\ldots,T\}} b_{k,t}^*+\tau \max\limits_{t\in\mathcal{W}(k)}\tilde{d}_{k,t}, \forall \tau \rightarrow 0^{+}$. Using this equation
and taking any $\omega(k)\in\arg\max\limits_{t\in\mathcal{W}(k)}\tilde{d}_{k,t}$}
we can simplify \eqref{first-order} as
\begin{eqnarray}\label{first-order2}
\sum_{m=1}^M \nabla f_m\left(\mathbf{b}^*\right)^\mathrm{T}\textcolor{black}{\tilde{\mathbf{d}}}
+\rho\sum_{k=1}^K \left(\sum_{t=0}^T\textcolor{black}{\tilde{d}}_{k,t}- \textcolor{black}{\tilde{d}}_{k,\omega(k)}\right) \geq 0,\nonumber\\
\forall \textcolor{black}{\tilde{\mathbf{d}}}\in\mathcal{T}\left(\mathbf{b}^*\right).
\end{eqnarray}
\textcolor{black}{Notice that} the first term of \eqref{first-order2} can be upper bounded as
\begin{equation}\label{cs}
\sum_{m=1}^M \nabla f_m\left(\mathbf{b}^*\right)^\mathrm{T}\textcolor{black}{\tilde{\mathbf{d}}}
\leq
\textcolor{black}{\left\Vert  \sum_{m=1}^M \nabla f_m\left(\mathbf{b}^*\right)\right\Vert_{\infty}
\left\Vert \textcolor{black}{\tilde{\mathbf{d}}} \right\Vert_{1}}
\triangleq \rho^{*}\left\Vert \textcolor{black}{\tilde{\mathbf{d}}} \right\Vert_{1},
\end{equation}
where $\rho^*<\infty$ since $\mathbf{b}^*$ is bounded in $[0,1]^{K(T+1)}$.
Substituting \eqref{cs} into \eqref{first-order2}, we obtain
\begin{equation}\label{cs2}
\rho^{*}\left\Vert \textcolor{black}{\tilde{\mathbf{d}}} \right\Vert_{1}
+\rho\sum_{k=1}^K \left(\sum_{t=0}^T\textcolor{black}{\tilde{d}}_{k,t}- \textcolor{black}{\tilde{d}}_{k,\omega(k)}\right) \geq 0,~~
\forall \textcolor{black}{\tilde{\mathbf{d}}}\in\mathcal{T}\left(\mathbf{b^*}\right).
\end{equation}

Based on \eqref{cs2}, we prove that when $\rho>\rho^*$,
$\mathbf{b}^*$ must satisfy \eqref{opt1c2} by contradiction.
Suppose that \eqref{opt1c2} does not hold, i.e., there exits $\bar{k}\in\left\{1,\ldots,K\right\}$
such that
$\left\Vert\mathbf{b}_{\bar{k}}^*\right\Vert_0>1$.
Moreover, since $\mathbf{b}^*\in[0,1]^{K(T+1)}$ in problem \eqref{opt2},
there must exist one $\phi(\bar{k})\neq \omega(\bar{k})$ satisfying $b_{\bar{k},\phi(\bar{k})}^*>0$.
On the other hand, notice that
any $\textcolor{black}{\tilde{\mathbf{d}}}\in\mathcal{T}\left(\mathbf{b^*}\right)$ should
satisfy
\begin{equation}\label{tagent}
\begin{cases}
\textcolor{black}{\tilde{d}}_{k,t} \geq 0, &\text{ if } b_{k,t}^*=0,\cr
\textcolor{black}{\tilde{d}}_{k,t} \in \mathbb{R}, &\text{ if } b_{k,t}^*\in(0,1),\cr
\textcolor{black}{\tilde{d}}_{k,t} \leq 0, &\text { if } b_{k,t}^*=1,
\end{cases}~~\forall k= 1,\ldots,K,~~\forall t= 0,\ldots,T.
\end{equation}
Consequently,
there exists one $\textcolor{black}{\tilde{\mathbf{d}}}\in\mathcal{T}\left(\mathbf{b^*}\right)$ satisfying
$\textcolor{black}{\tilde{d}}_{\bar{k},\phi(\bar{k})}=-1$ and $\textcolor{black}{\tilde{d}}_{k,t}=0, \forall (k,t)\neq (\bar{k}, \phi(\bar{k}))$.
Substituting this $\textcolor{black}{\tilde{\mathbf{d}}}$ into \eqref{cs2}, we obtain
$\rho^{*}-\rho\geq 0$, which is contradictory to $\rho>\rho^*$.
Therefore, when $\rho>\rho^*$, $\mathbf{b}^*$ must satisfy \eqref{opt1c2}.
Together with the fact that \eqref{opt1c1} is already satisfied in problem~\eqref{opt2},
we can conclude that
$\mathbf{b}^*$ is a feasible point of problem~\eqref{opt1}.

Next, we prove that $\mathbf{b}^*$ must also be a stationary point problem \eqref{opt1} when $\rho>\rho^*$.
Since $\mathbf{b}^*$ has been proved to be feasible to
problem \eqref{opt1}, we have
$\left\Vert\mathbf{b}_k^*\right\Vert_0\leq1$.
Moreover, with $\mathcal{P}\left(\mathbf{b^*}\right)$
denoting the tangent cone of the feasible set of problem \eqref{opt1} at $\mathbf{b^*}$,
for any $\textcolor{black}{\tilde{\mathbf{d}}}\in\mathcal{P}\left(\mathbf{b^*}\right)$,
we have
$\left\Vert\mathbf{b}_k^*+\tau\textcolor{black}{\tilde{\mathbf{d}}}_k\right\Vert_0\leq1, \forall \tau\rightarrow0^+, \forall k=1,\ldots,K$.
Consequently, \textcolor{black}{we have
$\sum\limits_{t=0}^T\left(b_{k,t}^*+\tau
\tilde{d}_{k,t}\right)=\max\limits_{t\in\{0,\ldots,T\}}\left(b_{k,t}^*+\tau
\tilde{d}_{k,t}\right), \forall \tau\rightarrow0^+$. If $\left\Vert\mathbf{b}_k^*\right\Vert_0=0$, it reduces to $\sum\limits_{t=0}^T\textcolor{black}{\tilde{d}}_{k,t}=\tilde{d}_{k,\omega(k)}$.
If $\left\Vert\mathbf{b}_k^*\right\Vert_0=1$,
from $\left\Vert\mathbf{b}_k^*+\tau\textcolor{black}{\tilde{\mathbf{d}}}_k\right\Vert_0\leq1$,
we have $\tilde{d}_{k,t}=0, \forall t\neq\omega(k)$.
Thus, it also follows that
$\sum\limits_{t=0}^T\textcolor{black}{\tilde{d}}_{k,t}=\tilde{d}_{k,\omega(k)}$. On the other hand,
since the feasible set of \eqref{opt1} is a subset of the feasible set of \eqref{opt2},
we have $\mathcal{P}\left(\mathbf{b^*}\right)\subseteq\mathcal{T}\left(\mathbf{b^*}\right)$.
Substituting
$\sum\limits_{t=0}^T\textcolor{black}{\tilde{d}}_{k,t}=\tilde{d}_{k,\omega(k)}$}
into \eqref{first-order2}
and focusing on $\textcolor{black}{\tilde{\mathbf{d}}}\in\mathcal{P}\left(\mathbf{b^*}\right)$ give
\begin{equation}\label{first-order4}
\sum_{m=1}^M \nabla f_m\left(\mathbf{b}^*\right)^\mathrm{T}\textcolor{black}{\tilde{\mathbf{d}}}
 \geq 0,~~
\forall \textcolor{black}{\tilde{\mathbf{d}}}\in\mathcal{P}\left(\mathbf{b^*}\right),
\end{equation}
which means that $\mathbf{b^*}$ is also a stationary point of problem \eqref{opt1}.

\textcolor{black}{Finally, we prove that the global optimal solutions
of the two problems are identical. Let $\mathbf{b}^*$
denote the global optimal solution of problem \eqref{opt2}.
Let $F(\cdot)$ and $\mathbf{b}^{\star}$
denote the cost function and the optimal solution of problem \eqref{opt1}, respectively.}
Since both $\mathbf{b}^*$ and $\mathbf{b}^{\star}$ are feasible to
problem~\eqref{opt1}, we have
$\left\Vert\mathbf{b}_k^*\right\Vert_0\leq1$ and
$\left\Vert\mathbf{b}_k^\star\right\Vert_0\leq1$.
Consequently, we have $\sum\limits_{t=0}^Tb_{k,t}^*=\max\limits_{t\in\{0,\ldots,T\}}b_{k,t}^*$ and
$\sum\limits_{t=0}^Tb_{k,t}^\star=\max\limits_{t\in\{0,\ldots,T\}}b_{k,t}^\star$.
With $G(\cdot)$ denoting the cost function of
problem \eqref{opt2}, by substituting the above two equalities
into $G\left(\mathbf{b}^*\right)$ and $G\left(\mathbf{b}^\star\right)$,
we have $F\left(\mathbf{b}^*\right)=G\left(\mathbf{b}^*\right)$ and $G\left(\mathbf{b}^{\star}\right)=F\left(\mathbf{b}^{\star}\right)$.
Moreover, since $\mathbf{b}^{*}$ is the optimal solution of problem \eqref{opt2},
we have $F\left(\mathbf{b}^*\right)=G\left(\mathbf{b}^*\right)\leq G\left(\mathbf{b}^{\star}\right)=F\left(\mathbf{b}^{\star}\right)$.
On the other hand, since $\mathbf{b}^*$ is a feasible point of problem \eqref{opt1},
we also have $F\left(\mathbf{b}^*\right)\geq F\left(\mathbf{b}^{\star}\right)$.
Combining the above two inequalities, we finally have
$F\left(\mathbf{b}^*\right)=G\left(\mathbf{b}^*\right)=G\left(\mathbf{b}^{\star}\right)=F\left(\mathbf{b}^{\star}\right)$,
which means that optimal solutions of problems \eqref{opt1} and \eqref{opt2} are identical.

\section{Proof of Proposition \ref{update2}}\label{proof of update2}
Problem \eqref{PG} can be decomposed into $K$ subproblems, with each written as
\begin{equation}\label{PG2}
\mathbf{b}_k^{(i)}
= \arg \min _{\mathbf{b}_k \in [0,1]^{T+1}}
\sum_{t=0}^T
\left(b_{k,t}-\alpha_{k,t}^{(i)}\right)^{2}
-2\eta_{i}\rho\max_{t\in\{0,\ldots,T\}}b_{k,t},
\end{equation}
where $\alpha_{k,t}^{(i)}\triangleq b_{k,t}^{(i-1)}-\eta_{i} d_{k,t}^{(i-1)}$.
Taking any $\omega(k)\in\arg\max\limits_{t\in\{0,\ldots,T\}}b_{k,t}^{(i)}$,
problem \eqref{PG2} can be equivalently written as
\begin{equation}\label{PG3}
b_{k,t}^{(i)}=
\begin{cases}
\arg \min\limits_{b_{k,t} \in [0,1]}
\left(b_{k,t}-\alpha_{k,t}^{(i)}\right)^{2}-2\eta_{i}\rho b_{k,t}, &\text {if } t=\omega(k), \cr
\arg \min\limits_{b_{k,t} \in [0,1]}
\left(b_{k,t}-\alpha_{k,t}^{(i)}\right)^{2},  &\text {otherwise},
\end{cases}
\end{equation}
which can be simplified as
\begin{equation}\label{closed2}
b_{k,t}^{(i)}=
\begin{cases}
\textcolor{black}{\Pi_{[0,1]}}\left(\alpha_{k,t}^{(i)}+\eta_{i}\rho\right), &\text {if } t=\omega(k), \cr
\textcolor{black}{\Pi_{[0,1]}}\left(\alpha_{k,t}^{(i)}\right),  &\text {otherwise}.
\end{cases}
\end{equation}

However, since $\omega(k)$ is defined based on the solution $\mathbf{b}_k^{(i)}$ itself,
\eqref{closed2} cannot be directly used as the closed-form solution.
Next, we prove
that $\omega(k)$ can be replaced by $\tau(k)\in\arg\max\limits_{t\in\{0,\ldots,T\}}\alpha_{k,t}^{(i)}$.
According to the definition of $\tau(k)$, we have
$\alpha_{k,\tau(k)}^{(i)}\geq\alpha_{k,t}^{(i)}, \forall t=0,\ldots,T$.
Since $\eta_{i}\rho>0$ and $\textcolor{black}{\Pi_{[0,1]}}\left(\cdot\right)$ is a monotonically non-decreasing function,
we have
\begin{equation}\label{ineq2}
\textcolor{black}{\Pi_{[0,1]}}\left(\alpha_{k,\tau(k)}^{(i)}+\eta_{i}\rho\right)\geq
\textcolor{black}{\Pi_{[0,1]}}\left(\alpha_{k,t}^{(i)}\right),~~\forall t=0,\ldots,T.
\end{equation}
On the other hand, from \eqref{closed2}, we also have
\begin{eqnarray}\label{ineq3}
b_{k,\omega(k)}^{(i)}=\textcolor{black}{\Pi_{[0,1]}}\left(\alpha_{k,\omega(k)}^{(i)}+\eta_{i}\rho\right)\geq
\textcolor{black}{\Pi_{[0,1]}}\left(\alpha_{k,t}^{(i)}\right)=b_{k,t}^{(i)},\nonumber\\
\forall t=0,\ldots,T.~~
\end{eqnarray}
Comparing \eqref{ineq2} and \eqref{ineq3}, we can conclude that $\omega(k)$ in \eqref{closed2}
can be replaced by $\tau(k)$, and
the optimal solution of \eqref{PG} is given by \eqref{closed}.

\section{Proof of Theorem \ref{stationary}}\label{proof of stationary}
Denote the cost function of problems \eqref{opt2} and \eqref{PG}
as $G\left(\mathbf{b}\right)=G_0\left(\mathbf{b}\right)-\rho\sum\limits_{k=1}^K\max\limits_{t\in\{0,\ldots,T\}}b_{k,t}$ and
$U\left(\mathbf{b}\right)$, respectively.
Taking the second-order Taylor expansion to $G_0\left(\mathbf{b}\right)$ at $\mathbf{b}^{(i-1)}$,
we have \eqref{taylor},
\begin{figure*}[!t]
\normalsize
\begin{eqnarray}\label{taylor}
G\left(\mathbf{b}\right)&\leq& \underbrace{G_0\left(\mathbf{b}^{(i-1)}\right)+\left(\mathbf{d}^{(i-1)}\right)^\mathrm{T}
\left(\mathbf{b}-\mathbf{b}^{(i-1)}\right)+\frac{L_{\text{d}}}{2}\left\Vert\mathbf{b}-\mathbf{b}^{(i-1)}\right\Vert_2^2
-\rho\sum_{k=1}^K\max_{t\in\{0,\ldots,T\}}b_{k,t}}_{U_1(\mathbf{b})},\nonumber\\
&\leq& \underbrace{G_0\left(\mathbf{b}^{(i-1)}\right)+\left(\mathbf{d}^{(i-1)}\right)^\mathrm{T}
\left(\mathbf{b}-\mathbf{b}^{(i-1)}\right)+\frac{1}{2\eta_i}\left\Vert\mathbf{b}-\mathbf{b}^{(i-1)}\right\Vert_2^2
-\rho\sum_{k=1}^K\max_{t\in\{0,\ldots,T\}}b_{k,t}}_{U_2\left(\mathbf{b}\right)},\nonumber\\
&=& U\left(\mathbf{b}\right)+C_0^{(i-1)},~~\forall \mathbf{b}\in[0,1]^{K(T+1)},
\end{eqnarray}
\hrulefill
\vspace*{4pt}
\end{figure*}
where $\mathbf{d}^{(i-1)}\triangleq \nabla G_0\left(\mathbf{b}^{(i-1)}\right)$,
$L_{\text{d}}$ is the Lipschitz constant of $\nabla G_0\left(\mathbf{b}\right)$,
and $C_0^{(i-1)}\triangleq G_0\left(\mathbf{b}^{(i-1)}\right)-\eta_i\left\Vert\mathbf{d}^{(i-1)}\right\Vert_2^2/2$.
The first inequality \textcolor{black}{in \eqref{taylor}} is due to the
\textcolor{black}{fact that $\nabla G_0(\mathbf{b})$ is Lipschitz continuous with constant $L_{\text{d}}$}
and the second inequality comes from $\eta_i<1/L_{\text{d}}$.
Consequently, we have
\begin{eqnarray}\label{df}
&&G\left(\mathbf{b}^{(i-1)}\right)-G\left(\mathbf{b}^{(i)}\right)\nonumber\\
&\geq& G\left(\mathbf{b}^{(i-1)}\right)-U_1\left(\mathbf{b}^{(i)}\right)=U_2\left(\mathbf{b}^{(i-1)}\right)-U_1\left(\mathbf{b}^{(i)}\right)\nonumber\\
&\geq& U_2\left(\mathbf{b}^{(i)}\right)-U_1\left(\mathbf{b}^{(i)}\right)\nonumber\\&=&
\left(\frac{1}{2\eta_i}-\frac{L_{\text{d}}}{2}\right)\left\Vert\mathbf{b}^{(i)}-\mathbf{b}^{(i-1)}\right\Vert_2^2\geq 0,
\end{eqnarray}
where the second inequality is because $\mathbf{b}^{(i)}$ is the optimal solution of problem \eqref{PG}
and hence also the minimizer of $U_2\left(\mathbf{b}\right)$.
Summing \eqref{df} over $i$ yields
\begin{eqnarray}\label{sumdf}
&&\sum_{i=1}^{\infty}\left(\frac{1}{2\eta_i}-\frac{L_{\text{d}}}{2}\right)\left\Vert\mathbf{b}^{(i)}-\mathbf{b}^{(i-1)}\right\Vert_2^2
\nonumber\\&\leq&
G\left(\mathbf{b}^{(0)}\right)-G\left(\mathbf{b}^{(\infty)}\right)<\infty,
\end{eqnarray}
where the second inequaltiy is becasue $G\left(\mathbf{b}\right)$ is bounded below by $\sum_{m=1}^M(L+T)\log\sigma_m^2$.
Since ${1}/{(2\eta_i)}-{L_{\text{d}}}/{2}>0$, \eqref{sumdf} yields
$\lim\limits_{i\rightarrow \infty}\left\Vert\mathbf{b}^{(i)}-\mathbf{b}^{(i-1)}\right\Vert_2=0$.
Thus, with $\mathbf{b}^{*}$ denoting a limit point of $\mathbf{b}^{(i)}$,
we also have $\mathbf{b}^{(i-1)}\rightarrow\mathbf{b}^{*}$.
Letting $i\rightarrow \infty$ in \eqref{PG}, we obtain \eqref{PGlim},
\begin{figure*}[!t]
\normalsize
\begin{eqnarray}\label{PGlim}
\mathbf{b}^{*}
&=& \arg \min _{\mathbf{b} \in [0,1]^{K(T+1)} } \frac{1}{2 \eta_{*}}\left\Vert \mathbf{b}-\left(\mathbf{b}^{*}-\textcolor{black}{\eta_*}\mathbf{d}^{*}\right)\right\Vert_{2}^{2}
-\rho\sum_{k=1}^K\max_{t\in\{0,\ldots,T\}}b_{k,t},\nonumber\\
&=&\arg \min _{\mathbf{b} \in [0,1]^{K(T+1)} }
G_0\left(\mathbf{b}^{*}\right)+\left(\mathbf{d}^{*}\right)^\mathrm{T}
\left(\mathbf{b}-\mathbf{b}^{*}\right)+\frac{1}{2\eta_*}\left\Vert\mathbf{b}-\mathbf{b}^{*}\right\Vert_2^2
-\rho\sum_{k=1}^K\max_{t\in\{0,\ldots,T\}}b_{k,t},~~
\end{eqnarray}
\hrulefill
\vspace*{4pt}
\end{figure*}
where \textcolor{black}{$\eta_*>0$ is the limit point of $\eta_i$ and}
the second equality comes from the last equality in \eqref{taylor}.
Since $\mathbf{b}^{*}$ is the optimal solution of \eqref{PGlim},
we have
\begin{eqnarray}\label{mini}
&&\left(\mathbf{d}^{*}\right)^\mathrm{T}
\left(\mathbf{b}-\mathbf{b}^{*}\right)+\frac{1}{2\eta_*}\left\Vert\mathbf{b}-\mathbf{b}^{*}\right\Vert_2^2-\rho\sum_{k=1}^K\max_{t\in\{0,\ldots,T\}}b_{k,t}
\nonumber\\
&\geq&-\rho\sum_{k=1}^K\max_{t\in\{0,\ldots,T\}}b_{k,t}^*,
\forall\mathbf{b}\in[0,1]^{K(T+1)}.
\end{eqnarray}
Let $\mathcal{T}\left(\mathbf{b^*}\right)$ denote the tangent cone of $\mathbf{b}\in[0,1]^{K(T+1)}$ at $\mathbf{b}^{*}$.
Substituting $\mathbf{b}=\mathbf{b}^{*}+\tau  \textcolor{black}{\tilde{\mathbf{d}}}$
with $\textcolor{black}{\tilde{\mathbf{d}}}\in\mathcal{T}\left(\mathbf{b^*}\right)$ into \eqref{mini}
and dividing $\tau$ on both sides of \eqref{mini},
we have
\begin{eqnarray}\label{first-order3}
&&-\rho\sum_{k=1}^K
\lim _{\tau \rightarrow 0^{+}} \frac{\max\limits_{t\in\{0,\ldots,T\}}\left(b_{k,t}^{*}+\tau  \textcolor{black}{\tilde{d}}_{k,t}\right)
- \max\limits_{t\in\{0,\ldots,T\}}b_{k,t}^*}{\tau}\nonumber\\
&&+\left(\mathbf{d}^{*}\right)^\mathrm{T} \textcolor{black}{\tilde{\mathbf{d}}}
\geq 0,~~
\forall  \textcolor{black}{\tilde{\mathbf{d}}}\in\mathcal{T}\left(\mathbf{b^*}\right).
\end{eqnarray}
Noticing that
$\mathbf{d}^{*}$ is the gradient of \textcolor{black}{$G_0(\mathbf{b})$},
we can rewrite \eqref{first-order3} as \eqref{first-order},
which means that the limit point $\mathbf{b^*}$ is a stationary point of problem \eqref{opt2}.

\section{Proof of Theorem \ref{stationary2}}\label{proof of stationary2}
\textcolor{black}{We first show that $\mathcal{L}\left(\left\{\mathbf{x}_m^{(i)}\right\}_{m=1}^M, \mathbf{b}^{(i)};
\left\{\boldsymbol{\lambda}_m^{(i)}\right\}_{m=1}^M
\right)$ is monotonically decreasing as $i$ increases.}
With $U_3\left(\mathbf{b}\right)\triangleq
\mathcal{L}\left(\left\{\mathbf{x}_m^{(i-1)}\right\}_{m=1}^M, \mathbf{b};
\left\{\boldsymbol{\lambda}_m^{(i-1)}\right\}_{m=1}^M
\right)+{\delta}/{2}\left\Vert\mathbf{b}-\mathbf{b}^{(i-1)}\right\Vert_2^2$
and $U_4\left(\mathbf{b}\right)$ denoting the cost function of problem \eqref{ADMM1},
we have
$\mathcal{L}\left(\left\{\mathbf{x}_m^{(i-1)}\right\}_{m=1}^M, \mathbf{b};
\left\{\boldsymbol{\lambda}_m^{(i-1)}\right\}_{m=1}^M
\right)\leq U_3\left(\mathbf{b}\right)=U_4\left(\mathbf{b}\right)+\sum\limits_{m=1}^M
f_m\left(\mathbf{x}_m^{(i-1)}\right)$,
where the equality holds when $\mathbf{b}=\mathbf{b}^{(i-1)}$.
Consequently,
\begin{eqnarray}\label{L1}
&&\mathcal{L}\left(\left\{\mathbf{x}_m^{(i-1)}\right\}_{m=1}^M, \mathbf{b}^{(i-1)};
\left\{\boldsymbol{\lambda}_m^{(i-1)}\right\}_{m=1}^M
\right)\nonumber\\
&&-\mathcal{L}\left(\left\{\mathbf{x}_m^{(i-1)}\right\}_{m=1}^M, \mathbf{b}^{(i)};
\left\{\boldsymbol{\lambda}_m^{(i-1)}\right\}_{m=1}^M
\right)\nonumber\\
&=&
U_3\left(\mathbf{b}^{(i-1)}\right)-\mathcal{L}\left(\left\{\mathbf{x}_m^{(i-1)}\right\}_{m=1}^M, \mathbf{b}^{(i)};
\left\{\boldsymbol{\lambda}_m^{(i-1)}\right\}_{m=1}^M
\right)\nonumber\\
&\geq&
U_3\left(\mathbf{b}^{(i)}\right)-\mathcal{L}\left(\left\{\mathbf{x}_m^{(i-1)}\right\}_{m=1}^M, \mathbf{b}^{(i)};
\left\{\boldsymbol{\lambda}_m^{(i-1)}\right\}_{m=1}^M
\right)\nonumber\\
&=&\frac{\delta}{2}\left\Vert\mathbf{b}^{(i)}-\mathbf{b}^{(i-1)}\right\Vert_2^2,
\end{eqnarray}
where the last inequality holds because $\mathbf{b}^{(i)}$
is the optimal solution of problem \eqref{ADMM1}, and hence also the minimizer of
$U_3\left(\mathbf{b}\right)$.

On the other hand, since $\mu>2L_m$,
we have $\mu\mathbf{I}_{K(T+1)}-\nabla^2 f_m(\mathbf{x}_m)\succeq L_m\mathbf{I}_{K(T+1)}$,
which means that problem \eqref{ADMM2} is strongly convex with modulus $L_m$.
Thus, \textcolor{black}{line~6} of Algorithm~2 can solve
problem \eqref{ADMM2} to the optimal solution $\mathbf{x}_m^{(i)}$ and we have
\begin{eqnarray}\label{L2}
&&\mathcal{L}\left(\left\{\mathbf{x}_m^{(i-1)}\right\}_{m=1}^M, \mathbf{b}^{(i)};
\left\{\boldsymbol{\lambda}_m^{(i-1)}\right\}_{m=1}^M
\right)\nonumber\\
&&-\mathcal{L}\left(\left\{\mathbf{x}_m^{(i)}\right\}_{m=1}^M, \mathbf{b}^{(i)};
\left\{\boldsymbol{\lambda}_m^{(i-1)}\right\}_{m=1}^M
\right)\nonumber\\
&\geq&\sum_{m=1}^M\frac{L_m}{2}\left\Vert\mathbf{x}_m^{(i)}-\mathbf{x}_m^{(i-1)}\right\Vert_2^2.
\end{eqnarray}
Moreover, since $\mathbf{x}_m^{(i)}$ is the optimal solution of
problem \eqref{ADMM2}, we have
\begin{equation}\label{ADMM2opt}
\nabla f_m\left(\mathbf{x}_{m}^{(i)}\right)
+\boldsymbol{\lambda}_m^{(i-1)}
+{\mu}\left(\mathbf{x}_m^{(i)}-\mathbf{b}^{(i)}\right)=\mathbf{0},~~\forall m=1,\ldots,M.
\end{equation}
Substituting \eqref{dual} into \eqref{ADMM2opt},
we obtain
\begin{eqnarray}\label{ADMM2opt2}
\left\Vert \boldsymbol{\lambda}_m^{(i-1)}-\boldsymbol{\lambda}_m^{(i)}\right\Vert_2&=&
\left\Vert \nabla f_m\left(\mathbf{x}_{m}^{(i-1)}\right)-\nabla f_m\left(\mathbf{x}_{m}^{(i)}\right)\right\Vert_2
\nonumber\\
&\leq& L_m
\left\Vert \mathbf{x}_{m}^{(i-1)}-\mathbf{x}_{m}^{(i)}\right\Vert_2.
\end{eqnarray}
Consequently, we have
\begin{eqnarray}\label{L3}
&&\mathcal{L}\left(\left\{\mathbf{x}_m^{(i)}\right\}_{m=1}^M, \mathbf{b}^{(i)};
\left\{\boldsymbol{\lambda}_m^{(i-1)}\right\}_{m=1}^M
\right)
\nonumber\\&&-\mathcal{L}\left(\left\{\mathbf{x}_m^{(i)}\right\}_{m=1}^M, \mathbf{b}^{(i)};
\left\{\boldsymbol{\lambda}_m^{(i)}\right\}_{m=1}^M
\right)\nonumber\\
&=&
\sum_{m=1}^M
\left(\boldsymbol{\lambda}_m^{(i-1)}-\boldsymbol{\lambda}_m^{(i)}\right)^\mathrm{T}
\left(\mathbf{x}_m^{(i)}-\mathbf{b}^{(i)}\right)\nonumber\\
&=&
-\frac{1}{\mu}\sum_{m=1}^M\left\Vert\boldsymbol{\lambda}_m^{(i-1)}-\boldsymbol{\lambda}_m^{(i)}\right\Vert_2^2
\nonumber\\
&\geq&
-\frac{1}{\mu}\sum_{m=1}^ML_m^2\left\Vert \mathbf{x}_{m}^{(i-1)}-\mathbf{x}_{m}^{(i)}\right\Vert_2^2,
\end{eqnarray}
where the second equality comes from \eqref{dual}
and the last inequality is due to \eqref{ADMM2opt2}.
Combining \eqref{L1}, \eqref{L2}, and \eqref{L3}, we obtain
\begin{eqnarray}\label{L4}
&&\mathcal{L}\left(\left\{\mathbf{x}_m^{(i-1)}\right\}_{m=1}^M, \mathbf{b}^{(i-1)};
\left\{\boldsymbol{\lambda}_m^{(i-1)}\right\}_{m=1}^M
\right)\nonumber\\
&&-\mathcal{L}\left(\left\{\mathbf{x}_m^{(i)}\right\}_{m=1}^M, \mathbf{b}^{(i)};
\left\{\boldsymbol{\lambda}_m^{(i)}\right\}_{m=1}^M
\right)\nonumber\\
&\geq&
\sum_{m=1}^M\left(\frac{L_m}{2}-\frac{L_m^2}{\mu}\right)\left\Vert \mathbf{x}_{m}^{(i)}-\mathbf{x}_{m}^{(i-1)}\right\Vert_2^2\nonumber\\
&&+\frac{\delta}{2}\left\Vert\mathbf{b}^{(i)}-\mathbf{b}^{(i-1)}\right\Vert_2^2\geq0,
\end{eqnarray}
where the last inequality is due to $\mu>2L_m$ and $\delta>0$.
Therefore, $\mathcal{L}\left(\left\{\mathbf{x}_m^{(i)}\right\}_{m=1}^M, \mathbf{b}^{(i)};
\left\{\boldsymbol{\lambda}_m^{(i)}\right\}_{m=1}^M
\right)$ is monotonically decreasing.

Next, we prove that $\mathcal{L}\left(\left\{\mathbf{x}_m^{(i)}\right\}_{m=1}^M, \mathbf{b}^{(i)};
\left\{\boldsymbol{\lambda}_m^{(i)}\right\}_{m=1}^M
\right)$ is lower bounded.
Substituting \eqref{dual} into \eqref{ADMM2opt}, we have
$\boldsymbol{\lambda}_m^{(i)}=-\nabla f_m\left(\mathbf{x}_{m}^{(i)}\right)$.
Substituting this equation into \eqref{ALF},
we obtain
\begin{eqnarray}\label{ALF2}
&&\mathcal{L}\left(\left\{\mathbf{x}_m^{(i)}\right\}_{m=1}^M, \mathbf{b}^{(i)};
\left\{\boldsymbol{\lambda}_m^{(i)}\right\}_{m=1}^M
\right)\nonumber\\
&=&\sum_{m=1}^Mf_m\left(\mathbf{x}_{m}^{(i)}\right)
+\rho\sum_{k=1}^K \left(\sum_{t=0}^Tb_{k,t}^{(i)}-\max_{t\in\{0,\ldots,T\}}b_{k,t}^{(i)}\right)
\nonumber\\
&&+\sum_{m=1}^{M}\nabla f_m\left(\mathbf{x}_{m}^{(i)}\right)^\mathrm{T}\left(\mathbf{b}^{(i)}-\mathbf{x}_m^{(i)}\right)
\nonumber\\&&+\frac{\mu}{2}\sum_{m=1}^{M}\left\Vert\mathbf{x}_m^{(i)}-\mathbf{b}^{(i)}\right\Vert_2^2.
\end{eqnarray}
Since $\mu\mathbf{I}_{K(T+1)}\succeq L_m\mathbf{I}_{K(T+1)} \succeq
\nabla^2 f_m\left(\mathbf{x}_{m}\right) $, we have
\begin{eqnarray}\label{ALF3}
\sum_{m=1}^Mf_m\left(\mathbf{x}_{m}^{(i)}\right)
+\sum_{m=1}^{M}\nabla f_m\left(\mathbf{x}_{m}^{(i)}\right)^\mathrm{T}\left(\mathbf{b}^{(i)}-\mathbf{x}_m^{(i)}\right)
\nonumber\\+\frac{\mu}{2}\sum_{m=1}^{M}\left\Vert\mathbf{x}_m^{(i)}-\mathbf{b}^{(i)}\right\Vert_2^2
\geq
\sum_{m=1}^Mf_m\left(\mathbf{b}^{(i)}\right).
\end{eqnarray}
Substituting \eqref{ALF3} into \eqref{ALF2}, we obtain
\begin{eqnarray}\label{ALF4}
&&\mathcal{L}\left(\left\{\mathbf{x}_m^{(i)}\right\}_{m=1}^M, \mathbf{b}^{(i)};
\left\{\boldsymbol{\lambda}_m^{(i)}\right\}_{m=1}^M
\right)
\nonumber\\
&\geq&\sum_{m=1}^Mf_m\left(\mathbf{b}^{(i)}\right)
+\rho\sum_{k=1}^K \left(\sum_{t=0}^Tb_{k,t}^{(i)}-\max_{t\in\{0,\ldots,T\}}b_{k,t}^{(i)}\right)\nonumber\\
&\geq&\sum_{m=1}^M(L+T)\log\sigma_m^2,
\end{eqnarray}
which means that $\mathcal{L}\left(\left\{\mathbf{x}_m^{(i)}\right\}_{m=1}^M, \mathbf{b}^{(i)};
\left\{\boldsymbol{\lambda}_m^{(i)}\right\}_{m=1}^M
\right)$ is lower bounded.

\textcolor{black}{Finally, we prove that any limit point of the sequence
$\left(\left\{\mathbf{x}_m^{(i)}\right\}_{m=1}^M, \mathbf{b}^{(i)};
\left\{\boldsymbol{\lambda}_m^{(i)}\right\}_{m=1}^M
\right)$
is a stationary point of problem \eqref{opt3}.}
Summing \eqref{L4} over $i$ yields
\begin{eqnarray}\label{L5}
&&\sum_{i=1}^\infty\frac{\delta}{2}\left\Vert\mathbf{b}^{(i)}-\mathbf{b}^{(i-1)}\right\Vert_2^2
\nonumber\\&&+
\sum_{i=1}^\infty\sum_{m=1}^M\left(\frac{L_m}{2}-\frac{L_m^2}{\mu}\right)\left\Vert \mathbf{x}_{m}^{(i)}-\mathbf{x}_{m}^{(i-1)}\right\Vert_2^2\nonumber\\
&\leq&
\mathcal{L}\left(\left\{\mathbf{x}_m^{(0)}\right\}_{m=1}^M, \mathbf{b}^{(0)};
\left\{\boldsymbol{\lambda}_m^{(0)}\right\}_{m=1}^M
\right)\nonumber\\&&
-\mathcal{L}\left(\left\{\mathbf{x}_m^{(\infty)}\right\}_{m=1}^M, \mathbf{b}^{(\infty)};
\left\{\boldsymbol{\lambda}_m^{(\infty)}\right\}_{m=1}^M
\right)
<\infty,\nonumber\\
\end{eqnarray}
where the second inequality comes from \eqref{ALF4}.
Since $\mu>2L_m$ and $\delta>0$, \eqref{L5} yields
$\lim\limits_{i\rightarrow \infty}\left\Vert\mathbf{b}^{(i)}-\mathbf{b}^{(i-1)}\right\Vert_2=\lim\limits_{i\rightarrow \infty}\left\Vert\mathbf{x}_m^{(i)}-\mathbf{x}_m^{(i-1)}\right\Vert_2=0$. Together with \eqref{ADMM2opt2} and \eqref{dual},
we also have
$\lim\limits_{i\rightarrow \infty}\left\Vert\textcolor{black}{\boldsymbol{\lambda}_m^{(i)}-\boldsymbol{\lambda}_m^{(i-1)}}\right\Vert_2=\lim\limits_{i\rightarrow \infty}\left\Vert\mathbf{x}_m^{(i)}-\mathbf{b}^{(i)}\right\Vert_2=0$.
Thus, with $\left(\left\{\mathbf{x}_m^{*}\right\}_{m=1}^M, \mathbf{b}^{*};
\left\{\boldsymbol{\lambda}_m^{*}\right\}_{m=1}^M
\right)$ denoting a limit point of the sequence
$\left(\left\{\mathbf{x}_m^{(i)}\right\}_{m=1}^M, \mathbf{b}^{(i)};
\left\{\boldsymbol{\lambda}_m^{(i)}\right\}_{m=1}^M
\right)$,
we have
\begin{eqnarray}\label{s0}
\mathbf{b}^{(i-1)}\rightarrow\mathbf{b}^{*},~~
\mathbf{x}_m^{(i-1)}\rightarrow\mathbf{x}_m^{*},~~
\boldsymbol{\lambda}_m^{(i-1)}\rightarrow\boldsymbol{\lambda}_m^{*},\nonumber\\
\forall m=1,\ldots,M,
\end{eqnarray}
\begin{equation}\label{s1}
\mathbf{x}_m^{*}=\mathbf{b}^{*},~~\forall m=1,\ldots,M.
\end{equation}
Taking limit for \eqref{ADMM2opt}, and using \eqref{s0} and \eqref{s1},
we have
\begin{equation}\label{s2}
\nabla f_m\left(\mathbf{x}_{m}^{*}\right)
+\boldsymbol{\lambda}_m^{*}
=\mathbf{0},~~\forall m=1,\ldots,M.
\end{equation}
On the other hand, since
$\mathbf{b}^{(i)}$ is the optimal solution of problem \eqref{ADMM1}, it satisfies the following first-order
optimality condition:
\begin{eqnarray}\label{s3}
&&\rho\sum_{k=1}^K \sum_{t=0}^T\textcolor{black}{\tilde{d}}_{k,t}-\sum_{m=1}^{M}\left(\boldsymbol{\lambda}_m^{(i-1)}\right)^\mathrm{T}\textcolor{black}{\tilde{\mathbf{d}}}
\nonumber\\
&&-\rho\sum_{k=1}^K\lim_{\tau \rightarrow 0^{+}} \frac{\max\limits_{t\in\{0,\ldots,T\}}\left(b_{k,t}^{(i)}+\tau \textcolor{black}{\tilde{d}}_{k,t}\right)
- \max\limits_{t\in\{0,\ldots,T\}}b_{k,t}^{(i)}}{\tau}
\nonumber\\
&&+\delta\left(\mathbf{b}^{(i)}-\mathbf{b}^{(i-1)}\right)^\mathrm{T}\textcolor{black}{\tilde{\mathbf{d}}}
+\mu\sum_{m=1}^{M}\left(\mathbf{b}^{(i)}-\mathbf{x}_m^{(i-1)}\right)^\mathrm{T}\textcolor{black}{\tilde{\mathbf{d}}}
\geq 0,\nonumber\\
&&\forall \textcolor{black}{\tilde{\mathbf{d}}}\in\mathcal{T}\left(\mathbf{b}^{(i)}\right),
\end{eqnarray}
where $\mathcal{T}\left(\mathbf{b}^{(i)}\right)$ is the tangent cone of the feasible set
of problem \eqref{ADMM1} at $\mathbf{b}^{(i)}$.
Taking limit for \eqref{s3}, using \eqref{s0} and \eqref{s1},
and noticing that $\mathcal{T}\left(\mathbf{b}^{*}\right)\subseteq
\mathcal{T}\left(\mathbf{b}^{i}\right)$ from \eqref{tagent},
we have
\begin{eqnarray}\label{s4}
&&\rho\sum_{k=1}^K \sum_{t=0}^T\textcolor{black}{\tilde{d}}_{k,t}
-\sum_{m=1}^{M}\left(\boldsymbol{\lambda}_m^{*}\right)^\mathrm{T}\textcolor{black}{\tilde{\mathbf{d}}}
\nonumber\\
&&-\rho\sum_{k=1}^K
\lim _{\tau \rightarrow 0^{+}} \frac{\max\limits_{t\in\{0,\ldots,T\}}\left(b_{k,t}^{*}+\tau \textcolor{black}{\tilde{d}}_{k,t}\right)
- \max\limits_{t\in\{0,\ldots,T\}}b_{k,t}^{*}}{\tau}
\nonumber\\
&\geq& 0,~~\forall \textcolor{black}{\tilde{\mathbf{d}}}\in\mathcal{T}\left(\mathbf{b}^{*}\right).
\end{eqnarray}
Combining \eqref{s1}, \eqref{s2}, and \eqref{s4}, we can conclude that
$\left(\left\{\mathbf{x}_m^{*}\right\}_{m=1}^M, \mathbf{b}^{*};
\left\{\boldsymbol{\lambda}_m^{*}\right\}_{m=1}^M
\right)$ is a stationary point of problem~\eqref{opt3}.

\bibliographystyle{IEEEtran}
\bibliography{IEEEabrv,LiYang}

\begin{thebibliography}{10}
\providecommand{\url}[1]{#1}
\csname url@samestyle\endcsname
\providecommand{\newblock}{\relax}
\providecommand{\bibinfo}[2]{#2}
\providecommand{\BIBentrySTDinterwordspacing}{\spaceskip=0pt\relax}
\providecommand{\BIBentryALTinterwordstretchfactor}{4}
\providecommand{\BIBentryALTinterwordspacing}{\spaceskip=\fontdimen2\font plus
\BIBentryALTinterwordstretchfactor\fontdimen3\font minus
  \fontdimen4\font\relax}
\providecommand{\BIBforeignlanguage}[2]{{%
\expandafter\ifx\csname l@#1\endcsname\relax
\typeout{** WARNING: IEEEtran.bst: No hyphenation pattern has been}%
\typeout{** loaded for the language `#1'. Using the pattern for}%
\typeout{** the default language instead.}%
\else
\language=\csname l@#1\endcsname
\fi
#2}}
\providecommand{\BIBdecl}{\relax}
\BIBdecl

\bibitem{Bockelmann2016}
C.~Bockelmann, N.~Pratas, H.~Nikopour, K.~Au, T.~Svensson, C.~Stefanovic,
  P.~Popovski, and A.~Dekorsy, ``Massive machine-type communications in
  \protect{5G: Physical and MAC-layer} solutions,'' \emph{{IEEE} Commun. Mag.},
  vol.~54, no.~9, pp. 59--65, Sep. 2016.

\bibitem{LiuLiang2018}
L.~Liu, E.~G. Larsson, W.~Yu, P.~Popovski, \v{C}. Stefanovi\'{c}, and
  E.~de~Carvalho, ``Sparse signal processing for grant-free massive
  connectivity: \protect{A} future paradigm for random access protocols in the
  internet of things,'' \emph{{IEEE} Signal Process. Mag.}, vol.~35, no.~5, pp.
  88--99, Sep. 2018.

\bibitem{XM21}
X.~Chen, D.~W.~K. Ng, W.~Yu, E.~G. Larsson, N.~Al-Dhahir, and R.~Schober,
  ``Massive access for \protect{5G} and beyond,'' \emph{{IEEE} J. Sel. Areas
  Commun.}, vol.~39, no.~3, pp. 615--637, Sep. 2021.

\bibitem{LiuLiang2018-1}
L.~Liu and W.~Yu, ``Massive connectivity with massive \protect{MIMO-Part I}:
  Device activity detection and channel estimation,'' \emph{{IEEE} Trans.
  Signal Process.}, vol.~66, no.~11, pp. 2933--2946, Jun. 2018.

\bibitem{LiuLiang2018-2}
------, ``Massive connectivity with massive \protect{MIMO--Part II}: Achievable
  rate characterization,'' \emph{{IEEE} Trans. Signal Process.}, vol.~66,
  no.~11, pp. 2947--2959, Jun. 2018.

\bibitem{Chen2018}
Z.~Chen, F.~Sohrabi, and W.~Yu, ``Sparse activity detection for massive
  connectivity,'' \emph{{IEEE} Trans. Signal Process.}, vol.~66, no.~7, pp.
  1890--1904, Apr. 2018.

\bibitem{Ding19}
T.~Ding, X.~Yuan, and S.~C. Liew, ``Sparsity learning-based multiuser detection
  in grant-free massive-device multiple access,'' \emph{{IEEE} Trans. Wireless
  Commun.}, vol.~18, no.~7, pp. 3569--3582, Jul. 2019.

\bibitem{Ke20}
M.~Ke, Z.~Gao, Y.~Wu, X.~Gao, and R.~Schober, ``Compressive sensing-based
  adaptive active user detection and channel estimation: {Massive} access meets
  massive {MIMO},'' \emph{{IEEE} Trans. Signal Process.}, vol.~68, pp.
  764--779, 2020.

\bibitem{Chen2019}
Z.~Chen, F.~Sohrabi, and W.~Yu, ``Multi-cell sparse activity detection for
  massive random access: \protect{Massive MIMO} versus cooperative
  \protect{MIMO},'' \emph{{IEEE} Trans. Wireless Commun.}, vol.~18, no.~8, pp.
  4060--4074, Aug. 2019.

\bibitem{Senel2018}
K.~Senel and E.~G. Larsson, ``Grant-free massive \protect{MTC}-enabled massive
  \protect{MIMO: A} compressive sensing approach,'' \emph{{IEEE} Trans.
  Commun.}, vol.~66, no.~12, pp. 6164--6175, Dec. 2018.

\bibitem{Yuan20_2}
W.~Yuan, N.~Wu, A.~Zhang, X.~Huang, Y.~Li, and L.~Hanzo, ``Iterative receiver
  design for {FTN} signaling aided sparse code multiple access,'' \emph{{IEEE}
  Trans. Wireless Commun.}, vol.~19, no.~2, pp. 915--928, Feb. 2020.

\bibitem{Yuan20_5}
W.~Yuan, N.~Wu, Q.~Guo, D.~W.~K. Ng, J.~Yuan, and L.~Hanzo, ``Iterative joint
  channel estimation, user activity tracking, and data detection for {FTN-NOMA}
  systems supporting random access,'' \emph{{IEEE} Trans. Commun.}, vol.~68,
  no.~5, pp. 2963--2977, May 2020.

\bibitem{Jiang2020}
S.~Jiang, X.~Yuan, X.~Wang, C.~Xu, and W.~Yu, ``Joint user identification,
  channel estimation, and signal detection for grant-free \protect{NOMA},''
  \emph{{IEEE} Trans. Wireless Commun.}, vol.~19, no.~10, pp. 6960--6976, Oct.
  2020.

\bibitem{Mei21}
Y.~Mei, Z.~Gao, Y.~Wu, W.~Chen, J.~Zhang, D.~W.~K. Ng, and M.~Di~Renzo,
  ``Compressive sensing based joint activity and data detection for grant-free
  massive {IoT} access,'' \emph{{IEEE} Trans. Wireless Commun.}, vol.~21,
  no.~3, pp. 1851--1869, Mar. 2022.

\bibitem{Ai22}
W.~Chen, H.~Xiao, L.~Sun, and B.~Ai, ``Joint activity detection and channel
  estimation in massive {MIMO} systems with angular domain enhancement,''
  \emph{{IEEE} Trans. Wireless Commun.}, vol.~21, no.~5, pp. 2999--3011, May
  2022.

\bibitem{Liu2018}
X.~Liu, Y.~Shi, J.~Zhang, and K.~B. Letaief, ``Massive \protect{CSI}
  acquisition for dense cloud-\protect{RANs} with spatial-temporal dynamics,''
  \emph{{IEEE} Trans. Wireless Commun.}, vol.~17, no.~4, pp. 2557--2570, Apr.
  2018.

\bibitem{He2018}
Q.~He, T.~Q.~S. Quek, Z.~Chen, Q.~Zhang, and S.~Li, ``Compressive channel
  estimation and multi-user detection in \protect{C-RAN} with low-complexity
  methods,'' \emph{{IEEE} Trans. Wireless Commun.}, vol.~17, no.~6, pp.
  3931--3944, Jun. 2018.

\bibitem{Li2019}
Y.~Li, M.~Xia, and Y.-C. Wu, ``Activity detection for massive connectivity
  under frequency offsets via first-order algorithms,'' \emph{{IEEE} Trans.
  Wireless Commun.}, vol.~18, no.~3, pp. 1988--2002, Mar. 2019.

\bibitem{Shao2020}
X.~Shao, X.~Chen, and R.~Jia, ``A dimension reduction-based joint activity
  detection and channel estimation algorithm for massive access,'' \emph{{IEEE}
  Trans. Signal Process.}, vol.~68, no.~1, pp. 420--435, Jan. 2020.

\bibitem{LiTing22}
T.~Li, J.~Zhang, Z.~Yang, Z.~L. Yu, Z.~Gu, and Y.~Li, ``Dynamic user activity
  and data detection for grant-free {NOMA} via weighted $\ell_{2,1}$
  minimization,'' \emph{{IEEE} Trans. Wireless Commun.}, vol.~21, no.~3, pp.
  1638--1651, Mar. 2022.

\bibitem{Haghighatshoar2018}
S.~Haghighatshoar, P.~Jung, and G.~Caire, ``Improved scaling law for activity
  detection in massive \protect{MIMO} systems,'' in \emph{IEEE ISIT}, 2018.

\bibitem{chenIcc}
Z.~Chen, F.~Sohrabi, Y.-F. Liu, and W.~Yu, ``Covariance based joint activity
  and data detection for massive random access with massive \protect{MIMO},''
  in \emph{IEEE Int. Conf. Commun. (ICC)}, 2019.

\bibitem{icasspW}
Z.~Wang, Z.~Chen, Y.-F. Liu, F.~Sohrab, and W.~Yu, ``An efficient active set
  algorithm for covariance based joint data and activity detection for massive
  random access with massive \protect{MIMO},'' in \emph{IEEE International
  Conference on Acoustics, Speech and Signal Processing (ICASSP)}, 2021.

\bibitem{Wang21}
Z.~Wang, Y.-F. Liu, Z.~Chen, and W.~Yu, ``Accelerating coordinate descent via
  active set selection for device activity detection for multi-cell massive
  random access,'' in \emph{IEEE International Workshop on Signal Processing
  Advances in Wireless Communications (SPAWC)}, 2021.

\bibitem{Chen2021}
Z.~Chen, F.~Sohrabi, and W.~Yu, ``Sparse activity detection in multi-cell
  massive \protect{MIMO} exploiting channel large-scale fading,'' \emph{{IEEE}
  Trans. Signal Process.}, vol.~69, pp. 3768--3781, 2021.

\bibitem{Lin2022}
Q.~Lin, Y.~Li, and Y.-C. Wu, ``Sparsity constrained joint activity and data
  detection for massive access: {A} difference-of-norms penalty framework,''
  \emph{{IEEE} Trans. Wireless Commun.}, to appear 2022,
  doi:10.1109/TWC.2022.3204786.

\bibitem{Fengler2021}
A.~Fengler, S.~Haghighatshoar, P.~Jung, and G.~Caire, ``\protect{Non-Bayesian}
  activity detection, large-scale fading coefficient estimation, and unsourced
  random access with a massive \protect{MIMO} receiver,'' \emph{{IEEE} Trans.
  Inf. Theory}, vol.~67, no.~5, pp. 2925--2951, May 2021.

\bibitem{Chen2020}
Z.~Chen, F.~Sohrabi, Y.-F. Liu, and W.~Yu, ``Phase transition analysis for
  covariance based massive random access with massive \protect{MIMO},''
  \emph{{IEEE} Trans. Inf. Theory}, vol.~68, no.~3, pp. 1696--1715, Mar. 2022.

\bibitem{Ke2021}
M.~Ke, Z.~Gao, Y.~Wu, X.~Gao, and K.-K. Wong, ``Massive access in cell-free
  massive \protect{MIMO}-based internet of things: \protect{Cloud} computing
  and edge computing paradigms,'' \emph{{IEEE} J. Sel. Areas Commun.}, vol.~39,
  no.~3, pp. 756--772, Mar. 2021.

\bibitem{Shao2020-2}
X.~Shao, X.~Chen, D.~W.~K. Ng, C.~Zhong, and Z.~Zhang, ``Cooperative activity
  detection: \protect{Sourced} and unsourced massive random access paradigms,''
  \emph{{IEEE} Trans. Signal Process.}, vol.~68, pp. 6578--6593, 2020.

\bibitem{Ganesan2021}
U.~K. Ganesan, E.~Bj\"{o}rnson, and E.~G. Larsson, ``Clustering based activity
  detection algorithms for grant-free random access in cell-free massive
  \protect{MIMO},'' \emph{{IEEE} Trans. Commun.}, vol.~69, no.~11, pp.
  7520--7530, Nov. 2021.

\bibitem{YC11}
Y.-C. Wu, Q.~Chaudhari, and E.~Serpedin, ``Clock synchronization of wireless
  sensor networks,'' \emph{{IEEE} Signal Process. Mag.}, vol.~28, no.~1, pp.
  124--138, Jan. 2011.

\bibitem{Du2013}
J.~Du and Y.-C. Wu, ``Distributed clock skew and offset estimation in wireless
  sensor networks: Asynchronous algorithm and convergence analysis,''
  \emph{{IEEE} Trans. Wireless Commun.}, vol.~12, no.~11, pp. 5908--5917, Nov.
  2013.

\bibitem{Luo2013}
B.~Luo and Y.-C. Wu, ``Distributed clock parameters tracking in wireless sensor
  network,'' \emph{{IEEE} Trans. Wireless Commun.}, vol.~12, no.~12, pp.
  6464--6475, Dec. 2013.

\bibitem{Wang22}
Z.~Wang, Y.-F. Liu, and L.~Liu, ``Covariance-based joint device activity and
  delay detection in asynchronous {mMTC},'' \emph{{IEEE} Signal Process.
  Lett.}, vol.~29, pp. 538--542, Jan. 2022.

\bibitem{LiuLiang21}
L.~Liu and Y.-F. Liu, ``An efficient algorithm for device detection and channel
  estimation in asynchronous {IoT} systems,'' in \emph{IEEE International
  Conference on Acoustics, Speech and Signal Processing (ICASSP)}, 2021.

\bibitem{Foucart2009SparsestSO}
S.~Foucart and M.~J. Lai, ``Sparsest solutions of underdetermined linear
  systems via $\ell_q$-minimization for $0<q\leq1$,'' \emph{Applied and
  Computational Harmonic Analysis}, vol.~26, pp. 395--407, 2009.

\bibitem{Soubies15}
E.~Soubies, L.~Blanc-F{\'e}raud, and G.~Aubert, ``{A Continuous Exact $\ell_0$
  penalty (CEL0) for least squares regularized problem},'' \emph{{SIAM Journal
  on Imaging Sciences}}, vol.~8, no.~3, pp. 1607--1639, Jul. 2015.

\bibitem{facchinei2007finite}
F.~Facchinei and J.-S. Pang, \emph{Finite-dimensional variational inequalities
  and complementarity problems}.\hskip 1em plus 0.5em minus 0.4em\relax
  Springer Science \& Business Media, 2007.

\bibitem{Emil20}
E.~Bj\"{o}rnson and L.~Sanguinetti, ``Making cell-free massive {MIMO}
  competitive with {MMSE} processing and centralized implementation,''
  \emph{{IEEE} Trans. Wireless Commun.}, vol.~19, no.~1, pp. 77--90, Jan. 2020.

\bibitem{He21}
H.~He, X.~Yu, J.~Zhang, S.~Song, and K.~B. Letaief, ``Cell-free massive {MIMO}
  for {6G} wireless communication networks,'' \emph{Journal of Communications
  and Information Networks}, vol.~6, no.~4, pp. 321--335, Dec. 2021.

\bibitem{Manijeh19}
M.~Bashar, K.~Cumanan, A.~G. Burr, H.~Q. Ngo, M.~Debbah, and P.~Xiao, ``Max-min
  rate of cell-free massive {MIMO} uplink with optimal uniform quantization,''
  \emph{{IEEE} Trans. Commun.}, vol.~67, no.~10, pp. 6796--6815, Oct. 2019.

\bibitem{Manijeh21}
M.~Bashar, H.~Q. Ngo, K.~Cumanan, A.~G. Burr, P.~Xiao, E.~Bj\"{o}rnson, and
  E.~G. Larsson, ``Uplink spectral and energy efficiency of cell-free massive
  {MIMO} with optimal uniform quantization,'' \emph{{IEEE} Trans. Commun.},
  vol.~69, no.~1, pp. 223--245, Jan. 2021.

\bibitem{Mossberg2006}
M.~Mossberg, E.~K. Larsson, and E.~Mossberg, ``Estimation of large-scale fading
  channels from sample covariances,'' in \emph{Proceedings of the 45th IEEE
  Conference on Decision and Control}, 2006.

\bibitem{Wang2018}
C.~Wang, O.~Y. Bursalioglu, H.~Papadopoulos, and G.~Caire, ``On-the-fly
  large-scale channel-gain estimation for massive antenna-array base
  stations,'' in \emph{IEEE Int. Conf. Commun. (ICC)}, 2018.

\bibitem{BernardSklar97}
M.~Bashar, H.~Q. Ngo, K.~Cumanan, A.~G. Burr, P.~Xiao, E.~Bj\"{o}rnson, and
  E.~G. Larsson, ``Rayleigh fading channels in mobile digital communication
  systems part {I}: Characterization,'' \emph{{IEEE} Commun. Mag.}, vol.~35,
  no.~7, pp. 90--100, Jul. 1997.

\bibitem{Sadeghabadi18}
E.~Sadeghabadi, S.~M. Azimi-Abarghouyi, B.~Makki, and M.~Nasiri-Kenari,
  ``Asynchronous downlink massive {MIMO} networks: A stochastic geometry
  approach,'' 2022. [Online]. Available: https://arxiv.org/abs/1806.02953.

\bibitem{Fazel2003}
M.~Fazel, H.~Hindi, and S.~P. Boyd, ``Log-det heuristic for matrix rank
  minimization with applications to {Hankel} and euclidean distance matrices,''
  in \emph{Proceedings of the 2003 American Control Conference}, 2003.

\bibitem{Parikh2014}
N.~Parikh and S.~Boyd, ``Proximal algorithms,'' \emph{Found. Trends Optim.},
  vol.~1, no.~3, pp. 127--239, 2014.

\bibitem{Boyd2011}
S.~Boyd, N.~Parikh, E.~Chu, B.~Peleato, and J.~Eckstein, ``Distributed
  optimization and statistical learning via the alternating direction method of
  multipliers,'' \emph{Found. Trends Mach. Learn}, vol.~3, no.~1, pp. 1--122,
  2011.

\bibitem{3GPP17}
3GPP, ``Further advancements for {E-UTRA} physical layer aspects ({Release}
  9),'' Tech. Rep. 3GPP TS 36.814, Mar. 2017.

\end{thebibliography}

\end{document}